\documentclass[11pt]{article} 
\usepackage{amsmath,amssymb,amsfonts,amsthm}
\usepackage[margin=1in]{geometry}
\usepackage{aliascnt}
\usepackage{hyperref}
\usepackage{enumitem}
\usepackage{float} 
\usepackage{placeins}
\usepackage{mathtools} 
\usepackage{todonotes} 
\newcommand{\R}{\mathbb{R}}
\newcommand{\C}{\mathbb{C}}
 
\usepackage{graphicx} 
\theoremstyle{definition}
\newtheorem{theorem}{Theorem}[section] 
\newtheorem{lemma}[theorem]{Lemma} 

\newtheorem{proposition}[theorem]{Proposition}
\newtheorem{corollary}[theorem]{Corollary} 
\newaliascnt{remark}{theorem}
\newtheorem{remark}[remark]{Remark}
\aliascntresetthe{remark}

\DeclareMathOperator{\Spec}{Spec}
\DeclareMathOperator{\ind}{ind}
\DeclareMathOperator{\sign}{sign}

 \DeclareMathOperator{\APS}{APS}  \DeclareMathOperator{\wind}{wind} \DeclareMathOperator{\Dom}{Dom} \newcommand{\Z}{\mathbb{Z}} \DeclareMathOperator{\Ran}{Ran} 
\hypersetup{colorlinks,breaklinks, linkcolor = purple, citecolor = teal, urlcolor = purple, }
\numberwithin{equation}{section}

\title{Dirac Operators, APS Boundary Conditions, and Spectral Flow on a Finite Warped Cylinder}
\author{Taro Kimura$^*$, Sanchita Sharma%
\thanks{Université Bourgogne Europe, CNRS, IMB UMR5584, Dijon, France}} 
\date{}
\begin{document} 
\maketitle
\begin{abstract}
We study the Dirac operator on a finite warped cylinder coupled to a background \(U(1)\) gauge field. We identify the intrinsic endpoint operators defining the Atiyah--Patodi--Singer (APS) boundary condition and derive a determinant characterization of the modewise APS spectrum. In the constant-gauge, invertible setting, the endpoint reduced $\eta$ contributions cancel, so the APS index vanishes. For smooth gauge families, the APS projector becomes discontinuous when a boundary mode crosses zero. We therefore introduce a regularized APS-type family of self-adjoint endpoint conditions that remains continuous across such crossings. This regularized family admits a real-symplectic boundary formulation within the standard spectral-flow/Maslov framework: for nondegenerate regularization, the zero-mode set coincides with the boundary-zero set, and transverse boundary zeros give isolated regular crossings.
\end{abstract}
\setcounter{tocdepth}{2}
\tableofcontents 
\section{Introduction and Overview} 

In the context of Dirac-type operators on manifolds with boundary, the Atiyah-Patodi-Singer (APS) boundary condition~\cite{APS1,APS2,APS3} is a fundamental elliptic boundary condition that has far-reaching implications for various fields, including physics and mathematics. It is defined from the spectral decomposition of the induced self-adjoint boundary Dirac operator and enters both the APS index theorem and later spectral-flow and Maslov-index formulations; see, for example,~\cite{LawsonMichelsohn,BGV,Hijazi,BoossWojciechowski,BaerBallmann2012}. 

Recently, APS-type phenomena have also been revisited from the physical, bulk-boundary, and domain-wall viewpoints. In particular, Fukaya, Onogi, and Yamaguchi showed that the APS index can be recovered from a domain-wall Dirac operator in a physically natural setup~\cite{Fukaya}, and this perspective was subsequently developed further in a more systematic mathematical form in~\cite{FukayaMatsuo2020}; see also Fukaya's review~\cite{FukayaReview2021} for a broader overview of massive-fermion and domain-wall reformulations of index theory. Related bulk-boundary and anomaly-inflow interpretations appear in the work of Witten and Yonekura~\cite{WittenYonekura2019} and of Kobayashi and Yonekura~\cite{KobayashiYonekura2021}. In a two-dimensional setting, close in spirit to the present paper, Onogi and Yoda related the reformulated APS index to Berry-phase and bulk-boundary contributions in a domain-wall picture~\cite{OnogiYoda2021}. Complementary lattice studies of curved domain-wall fermions by Aoki, Fukaya, and collaborators are also relevant to the geometric setting considered here~\cite{AokiFukaya2022,AokiFukaya2023,AokiFukayaKan2024,AokiSingleCurved2024}. More recently, Zhu established a generalized APS/domain-wall formula without assuming the invertibility of the boundary Dirac operator~\cite{Zhu2023}, while Aoki et al.\ gave a lattice/domain-wall realization showing how APS-type index information can be captured discretely~\cite{Aoki2026}.

This paper studies the aforementioned structures on a finite warped cylinder coupled to a background \(U(1)\) field. 
In abstract form, these structures are well understood \cite{LawsonMichelsohn,BGV,Hijazi,BoossWojciechowski,BaerBallmann2012}. 
In explicit curved models, however, the bulk Dirac operator, the intrinsic boundary operator, the boundary \(\eta\)-terms, and the zero-mode crossing are less often written out simultaneously in a single concrete setting. 
In this model, we can compute the coupled Dirac operator explicitly and identify the boundary operators entering the APS projector \cite{BoossWojciechowski,LeschWojciechowski1996}; 
reduce the bulk spectra to ordinary differential equations and related characteristic/determinant formulations \cite{ScottWojciechowski2000}; 
prove cancellation of the endpoint reduced \(\eta\)-invariant contributions in the constant gauge case under the invertibility assumption \cite{Bunke1995,LeschWojciechowski1996,KirkLesch}; 
and formulate the parameter-dependent zero-mode problem in a symplectic boundary space, in the spirit of spectral flow and Maslov index theory \cite{RobbinSalamon1995,CappellLeeMiller1994,Nicolaescu1995,KirkLesch,vanDenDungenRonge2020,Prokhorova2013,GorokhovskyLesch2013}.

At the level of the reduced mode equations, the present model also sits in a broader analytic context. After Fourier decomposition, the warped-cylinder Dirac problem becomes a one-dimensional first-order coupled system, and hence, a scalar second-order equation; for the warped function considered here, this equation is of general Heun type. Reductions of Dirac equations in curved or warped geometries to Heun-type equations occur in several settings, including quantum-corrected and black-hole-type backgrounds~\cite{Baradaran2025,Filho2023}, one-dimensional Dirac systems with varying interaction profiles~\cite{Ishkhanyan2023}, and related recent analyses of Dirac systems via Heun equations~\cite{Loginov:2022yex,Blas:2025bkt,Blas2026}. Our focus, however, is the explicit warped-cylinder boundary-value problem and its APS and spectral consequences.

We work on a finite warped cylinder
\begin{equation}
M=[0,T]\times S^1,
\qquad
g=dt^2+f(t)^2\,d\theta^2,
\end{equation}
coupled to a background \(U(1)\) gauge field. This model is simple enough to admit explicit formulas, yet rich, exhibiting nontrivial APS boundary operators, modewise reduction of the bulk problem, reduced \(\eta\)-invariant contributions, and a parameter-dependent zero-mode crossing problem. For the warping function \(f(t)=e^t+\alpha e^{-t}\), the scalar mode equation is of general Heun type.

We explicitly work with the finite warped cylinder coupled to a background \(U(1)\) field, and compute the bulk Dirac operator, identify the intrinsic endpoint operators defining the APS boundary condition, characterize the modewise APS spectrum by a boundary determinant, and prove cancellation of the endpoint reduced \(\eta\)-invariant contributions in the constant gauge case under the invertibility assumption \(k+A\neq 0\) for all allowed modes. For parameter-dependent gauge families, we do not claim continuity of the APS projector across the boundary-zero set \(k+A(s)=0\). Instead, we introduce a regularized APS-type boundary family, continuous across those points, and use it to place the corresponding zero-mode problem into the standard spectral-flow/Maslov formalism.

\paragraph{Main results.}
\begin{enumerate}[label=(\roman*)]
\item For each Fourier mode \(k\), the coupled Dirac operator reduces to a one-dimensional first-order system; for warping function \(f(t)=e^t+\alpha e^{-t}\), the associated scalar second-order equation is of general Heun type.

\item The intrinsic self-adjoint boundary operators for APS condition are identified explicitly. For \(m=k+A\neq 0\),  we characterize the modewise APS spectrum by a boundary determinant condition \(F_k(\lambda)=0\).

\item In the constant gauge case, assuming \(k+A\neq 0\) for all allowed modes, the endpoint reduced \(\eta\)-invariant contributions cancel. Thus,
\begin{equation}
\ind(D^+_{\APS})=0.
\end{equation}

\item For smooth gauge families \(A=A(s)\), the zero-mode problem admits a real symplectic boundary formulation. For the continuous regularized APS-type boundary family introduced in Section~\ref{sec:maslov}, the corresponding zero-mode problem fits the standard spectral-flow/Maslov formalism. Under the assumption
\[
\delta \neq \frac{2}{\ell(T)},
\qquad
\ell(T)=\int_0^T \frac{d\tau}{f(\tau)}.
\]
The zero-mode set is exactly the boundary-zero set \(k+A(s)=0\), and isolated regular crossings are precisely the transverse boundary zeros.
\end{enumerate}

Sections~\ref{sec:2ddirac}-\ref{sec:aps-boundary} discuss the APS index. \autoref{sec:maslov} concerns the regularized APS-type family introduced only to study parameter variation across the boundary zeros \(k+A(s)=0\). All spectral-flow and Maslov statements in \autoref{sec:maslov} refer to the regularized family.

\paragraph{Organization of the paper.}
We organize the paper as follows. In \autoref{sec:2ddirac}, we derive the explicit Dirac operator on the warped cylinder, fix conventions for the spinor bundle and Clifford multiplication, and include the coupling to the background \(U(1)\) field. In \autoref{sec:bulk}, we perform the Fourier-mode reduction, derive the decoupled scalar equation, explain the reduction to Heun type for the specific warped function \(f(t)=e^t+\alpha e^{-t}\), and formulate the APS boundary condition together with the determinant spectral condition. In \autoref{sec:aps-boundary}, we study the chiral APS index, compute the relevant reduced \(\eta\)-invariants, and prove the cancellation of the two endpoint contributions in the constant gauge invertible assumption. In \autoref{sec:maslov}, we treat one-parameter gauge families, introduce the symplectic boundary formalism and the regularized APS-type Lagrangian family, and describe the resulting Maslov/spectral-flow framework together with explicit model paths \(A(s)\). The appendices discuss the longer ordinary differential equation reductions and auxiliary computations.

\paragraph{Acknowledgments.}
This work was supported by EIPHI Graduate School (No.~ANR-17-EURE-0002) and the Bourgogne-Franche-Comté region.

\section{Dirac operator on a finite two-dimensional warped manifold}\label{sec:2ddirac}
In this section, we derive the Dirac operator on a two-dimensional warped cylinder, fix the geometric and Clifford conventions, and include the coupling to a background $U(1)$ gauge field. Throughout, we work in the Euclidean signature.
Let 
\begin{equation}
M=[0,T]\times S^1,
\qquad T>0,
\end{equation}
be equipped with coordinates \((t,\theta)\), where \(t\in[0,T]\) and \(\theta\in[0,2\pi)\). We endow \(M\) with the warped product metric
\begin{equation}
g=dt^2+f(t)^2\,d\theta^2,
\end{equation}
where \(f\in C^\infty([0,T])\) satisfies \(f(t)>0\) for all \(t\in[0,T]\). Thus, \(g\) is a smooth Riemannian metric on \(M\). The boundary of \(M\) is
\begin{equation}
\partial M=\{0\}\times S^1\sqcup \{T\}\times S^1.
\end{equation}
We write
\begin{equation}
Y_0=\{0\}\times S^1,\qquad Y_T=\{T\}\times S^1.
\end{equation}
Hence, \(M\) is a compact Riemannian manifold with boundary.

Next, we fix an orthonormal frame on $M$. Following standard conventions, Greek indices $(\mu,\nu)$ refer to curved coordinates, whereas Latin indices $(a,b)$ refer to the orthonormal frame. Using $e^a=e^a_{\ \mu}\,dx^\mu$, we choose the orthonormal coframe
\begin{equation}
e^1=dt,\qquad e^2=f(t)\,d\theta,
\end{equation}
with the dual frame
\begin{equation}
e_1=\partial_t,\qquad e_2=\frac{1}{f(t)}\partial_\theta.
\end{equation}
The metric components in the coordinate basis are recovered from the vielbein by $g_{\mu \nu} = e^{a}_{\mu} e^{b}_{\nu} \delta_{ab}$.

This frame is adapted to the warped manifold: $e_1$ is normal to the boundary, while $e_2$ is the tangent to the circle, scaled by the warping factor. We orient $M$ by $dt\wedge d\theta$, that is, by the ordered orthonormal frame $(e_1,e_2)$. The inward unit normal is $\partial_t$ at $t=0$ and $-\partial_t$ at $t=T$. In \autoref{subsec:aps}, this sign change will induce the corresponding sign change in the boundary operators on $Y_0$ and $Y_T$ used in the APS projector. For brevity, we often write $f$ in place of $f(t)$.

\subsection{Spinor bundle and Dirac operator}\label{subsec:diracop}
Let $S \to  M$ denote the complex spinor bundle associated with the chosen spin structure on $M$.
The two-dimensional complex spinor bundle has rank two. Since \(M=[0,T]\times S^1\) is a cylinder, once a spin structure is fixed, the spinor bundle can be trivialized in the chosen orthonormal frame. Spinors can therefore be represented by \(\mathbb C^2\)-valued functions, with the choice of spin structure reflected in the periodic or anti-periodic boundary condition in the \(\theta\)-direction.

We choose a representation of the complex Clifford algebra \(\mathrm{Cl}_2\) in \(\mathbb C^2\) by Hermitian \(2\times 2\) matrices \(\gamma_1,\gamma_2\) satisfying
\begin{equation}
\{\gamma_a,\gamma_b\}=2\delta_{ab}.
\end{equation}
the associated Hermitian Dirac operator is,
\begin{equation}
D=i\,\gamma_a\nabla_{e_a},
\end{equation}
where \(i=\sqrt{-1}\). This convention is chosen so that the Clifford action of the induced boundary is defined below.
\begin{equation}
c(X)=-\,i\,\gamma(N)\gamma(X),\qquad X\in \mathcal{T}(\partial M),
\end{equation}
is Hermitian. Here, \(\mathcal{T}(\partial M)\) denotes the tangent vector bundle along the boundary of \(M\).

We represent the Clifford algebra using Pauli matrices, where $\gamma_1 = \sigma_1$  and $ \gamma_2 = \sigma_2$. With the orientation fixed above, we define the chirality matrix by
\begin{equation}
\gamma_3 = i \gamma_1 \gamma_2 = - \sigma_3.
\end{equation}
Therefore,
\begin{equation}
\gamma_1 = \begin{pmatrix} 0 & 1 \\ 1 & 0 \end{pmatrix}, \qquad
\gamma_2 = \begin{pmatrix} 0 & -i \\ i & 0 \end{pmatrix}, \qquad
\gamma_3 = \begin{pmatrix} -1 & 0 \\ 0 & 1 \end{pmatrix}.
\end{equation}

The spin connection on $S$ is determined by the Levi-Civita connection of $g$, with respect to the orthonormal frame $( e_1, e_2)$ introduced above. The only non-zero coefficients of the Levi-Civita connection are
\begin{equation}
\nabla_{e_{2}} e_1=  \frac{f^{\prime}}{f}e_2,  \quad  \nabla_{e_1} e_2 = - \frac{f^{\prime}}{f}e_1.
\end{equation}

Thus, the only independent non-vanishing spin connection one-form is
\begin{equation}
\omega^1_{ 2} = -\frac{f'}{f} e^2 = - \omega^2_{\ 1}.
\end{equation}
These spin connection one-forms satisfy Cartan's first structure equation $ d e^{a} + \omega^{a}_{b} \wedge e^{b} = 0$.
The spinor covariant derivative on \(S\) is denoted by \(\nabla^S\). In the chosen trivialization, it is given by
\begin{equation}
\nabla^S_{e_a}=\partial_{e_a}+\frac12\,\omega_{a\,bc}\Sigma^{bc},
\qquad
\Sigma^{bc}=\frac14[\gamma_b,\gamma_c],
\end{equation}
where \(\partial_{e_a}\) denotes directional differentiation along \(e_a\), acting componentwise on spinors. For the warped metric under consideration, this gives;
\begin{equation}
\nabla^S_{e_1} = \partial_t, \qquad \nabla^S_{e_2} = \frac{1}{f}\partial_\theta - \frac{1}{2} \frac{f'}{ f} \gamma_1 \gamma_2.
\end{equation}

\paragraph{Function spaces and sections.}
For a vector bundle \(V\to M\), we write \(\Gamma(V)\), \(L^2(M;V)\), and \(H^1(M;V)\) for the spaces of smooth, square-integrable, and first Sobolev sections of \(V\), defined using the Riemannian volume measure and the Hermitian bundle metric. The same notation applies to each boundary component \(Y_{t_0}\). On \([0,T]\), we use the standard scalar Sobolev spaces \(L^2([0,T])\) and \(H^1([0,T])\).

For each boundary component \(Y_{t_0}\), let \(S_{Y_{t_0}}\) denote the induced boundary spinor bundle. We then write
\begin{equation}
\Gamma(S_{Y_{t_0}}),\qquad
L^2(Y_{t_0};S_{Y_{t_0}}),\qquad
H^1(Y_{t_0};S_{Y_{t_0}})
\end{equation}
for the corresponding spaces of smooth, square-integrable, and first Sobolev boundary spinors. We also write
\begin{equation}
H^{1/2}(Y_{t_0};S_{Y_{t_0}})
\end{equation}
for the standard fractional Sobolev trace space on the boundary component \(Y_{t_0}\). Therefore, \(H^{1/2}(Y_{t_0};S_{Y_{t_0}})\) is the natural target of the trace map
\begin{equation}
H^1(M;S)\longrightarrow H^{1/2}(Y_{t_0};S_{Y_{t_0}})
\end{equation}
given by restriction of spinors to \(Y_{t_0}\).

\paragraph{Chirality splitting.}
The two-dimensional spinor bundle splits as follows.
\begin{equation}
S=S^+\oplus S^-,
\end{equation}
where, under the chosen trivialization,
\begin{equation}
S^+ \cong M\times \mathbb C\binom{0}{1},\qquad
S^- \cong M\times \mathbb C\binom{1}{0}.
\label{eq:chiral_split}
\end{equation}
$S^\pm$ are the $\pm1$-eigenbundles of the chirality operator $\gamma_3$. In our representation,
\begin{equation}
\gamma_3\binom{1}{0}=-\binom{1}{0},\qquad
\gamma_3\binom{0}{1}=+\binom{0}{1}.
\end{equation}

\subsection{Coupling to a background \(U(1)\) gauge field}\label{subsec:gaugetw}
Let \(E\to M\) be a Hermitian complex line bundle, and fix a global unitary trivialization in which the unitary connection takes the form
\begin{equation}
\nabla^E=d+iA\,d\theta,
\end{equation}
where \(A\in\mathbb R\) is constant. We then consider the spinor bundle \(S\otimes E\) coupled to the background gauge field and the corresponding Dirac operator
\begin{equation}
D:\Gamma(S\otimes E)\to \Gamma(S\otimes E).
\end{equation}
From this point onward, all spinor bundles, boundary spinor bundles, and the associated spaces of sections, Sobolev spaces, and \(L^2\)-spaces are understood to be coupled to \(E\). For brevity, we suppress \(E\) from the notation and continue to write \(S\), \(S^\pm\), and \(S_{Y_{t_0}}\) in place of \(S\otimes E\), \(S^\pm\otimes E\), and \(S_{Y_{t_0}}\otimes E|_{Y_{t_0}}\), respectively.
In gauge theory notation, the connection one-form is $ A_\mu dx^\mu = A\,d\theta$, so that $A_t=0$ and $A_{\theta}=A$ is constant. 

The Dirac operator is defined using the tensor product connection.
\begin{equation}
\nabla^{S\otimes E}=\nabla^S\otimes 1 + 1\otimes \nabla^E
\end{equation}
on $S\otimes E$. Thus,
\begin{equation}
D = i\gamma_a \nabla^{S\otimes E}_{e_a}.
\end{equation}
In this trivialization, one obtains
\begin{equation}
\nabla^{S\otimes E}_{e_1}=\partial_t,\qquad
\nabla^{S\otimes E}_{e_2}=\frac1{f}(\partial_\theta+iA)-\frac{f'}{2f}\gamma_1\gamma_2.
\end{equation}
Combining these expressions, we obtain the full Dirac operator:
\begin{equation}
D = i \gamma_1 \left( \partial_t + \frac{f'}{2 f} \right) + i \gamma_2 \frac{1}{f} (\partial_{\theta} + iA).
\end{equation}
With Hermitian Euclidean gamma matrices, \(D\) is self-adjoint. For the APS boundary condition, however, the relevant operator is not the tangential term appearing in the bulk decomposition, but the intrinsic self-adjoint boundary Dirac operator defined on each boundary component below.
After plugging in the Pauli matrices in the previous equation, we obtain 
\begin{equation} \label{eq:Diracop}
D = 
i \begin{pmatrix}
0 & \partial_t + \dfrac{f'}{2f} - \dfrac{i}{f}( \partial_{\theta} + i A) \\
\partial_t + \dfrac{f'}{2f} + \dfrac{i}{f} ( \partial_{\theta} + i A)& 0
\end{pmatrix}.
\end{equation} 
The term $\frac{f'}{2f}$ encodes the effects of warping geometry. 
Thus, we obtain the Dirac operator on the warped cylinder coupled to the background gauge field.

\section{Bulk spectrum}\label{sec:bulk}
We observe that the metric and background connection are invariant in the \(\theta\)-direction. Therefore, the operator decomposes into a one-dimensional radial equation parametrized by the angular mode \(k\). 

\begin{remark} \label{remark:spin_structure}
We allow either spin structure on \(S^1\): for the periodic spin structure, \(k\in\mathbb Z\), whereas for the anti-periodic spin structure, \(k\in\mathbb Z+\tfrac12\). In what follows, \(k\) denotes an arbitrary allowed mode; the derivation is identical in both cases.
\end{remark}

Thus, the bulk spectral equation reduces to a modewise family of radial equations parametrized by
\begin{equation}
m=k+A.
\end{equation}
We now consider the eigenvalue equation and the associated Hilbert space of eigenfunctions.
\begin{equation}
D\psi=\lambda\psi,
\qquad
\lambda\in\mathbb R.
\end{equation}
\paragraph{Mode Hilbert spaces.}
For a fixed Fourier mode $k$, the natural radial Hilbert space induced from
$L^2(M;S\otimes E)$ is
\begin{equation}
\mathfrak H_k=L^2([0,T],f(t)\,dt;\C^2),
\end{equation}
with the corresponding Sobolev space
\begin{equation}
\mathfrak W_k=H^1([0,T],f(t)\,dt;\C^2).
\end{equation}
Accordingly, all modewise adjointness and self-adjointness statements in this section are understood with respect to the weighted radial inner product unless an explicit unitary conjugation to the unweighted space is made later. Writing
\begin{equation}\label{eq: ansatz}
\psi(t,\theta) = 
e^{ik \theta} \begin{pmatrix}
u(t) \\
v(t)
\end{pmatrix},
\end{equation}
we have
\begin{equation}
i \begin{pmatrix}
0 & \partial_t + \dfrac{f'}{2f} + \dfrac{k + A}{f} \\
\partial_t + \dfrac{f'}{2f} - \dfrac{k + A}{f} & 0
\end{pmatrix}
\begin{pmatrix}
u(t) \\
v(t)
\end{pmatrix}
= \lambda
\begin{pmatrix}
u(t) \\
v(t)
\end{pmatrix} ,
\end{equation}
leading to coupled differential equations
\begin{subequations}\label{eq:uv_system}
\begin{align}
v' + \dfrac{f'}{2f} v + \dfrac{k + A}{f} v &= -i\lambda u, \label{eq:uv_system_a}\\
u' + \dfrac{f'}{2f} u - \dfrac{k + A}{f} u &= -i\lambda v. \label{eq:uv_system_b}
\end{align}
\end{subequations}
The coupled system decouples into second-order scalar equations for the two spinor components. Writing \(m=k+A\), we obtain for \(u(t)\)
\begin{equation}
u^{\prime \prime} + \dfrac{f^\prime}{f} u' + \left( \frac{f^{\prime \prime}}{2 f} - \frac{f'^2}{4f^2} + \frac{m f'}{f^2} - \frac{m^2}{f^2} + \lambda^2 \right) u = 0.
\end{equation}
The corresponding equation for \(v(t)\) is obtained by replacing \(m\) with \(-m\).

\subsection{Reduction to Heun type}\label{subsec:redtoheun}
To place the resulting second-order equation into a standard Fuchsian form, we perform a sequence of transformations. First, we apply a Liouville transformation, which removes the first-derivative term and rewrites the equation in Schr\"odinger form.
 In particular, if the decoupled equation is written as
\begin{equation}\label{eq:decoupled-general}
u''(t)+P(t)\,u'(t)+R(t)\,u(t)=0,
\end{equation}
then setting \(u(t)=e^{-\frac12\int^t P(s)\,ds}\,w(t)\) produces
\begin{equation}\label{eq:liouville}
w''(t)+\Bigl(R(t)-\tfrac12 P'(t)-\tfrac14 P(t)^2\Bigr)\,w(t)=0.
\end{equation}
In our case \(P(t)=f'(t)/f(t)\), hence \(u(t)=f(t)^{-1/2}w(t)\). The explicit effective potential obtained from \eqref{eq:liouville} is given in \autoref{app:heun}.

Now we substitute $f(t)=e^t+\alpha e^{-t}$ with $\alpha>0$, and introduce $z=e^t$. Since
\begin{equation}\label{eq:zchange}
f=\frac{z^2+\alpha}{z},\qquad \frac{d}{dt}=z\frac{d}{dz},
\end{equation}
on the original real domain $t\in[0,T]$, we have
$z=e^t\in[1,e^T]\subset\mathbb R_{>0}$. For singularity analysis, we consider the normalized equation as a complex ODE in the variable $z$. Its coefficients are rational functions of $z$, so the equation is regular on
\begin{equation}
\mathbb C\setminus\{0,\pm i\sqrt{\alpha}\},
\end{equation}
with singular points at
\begin{equation}\label{eq:zsing}
z\in\{0,\ \pm i\sqrt{\alpha},\ \infty\}.
\end{equation}

To place these four singular points in a standard configuration, let \(\beta=\sqrt{\alpha}\), and then we apply the M\"obius transformation
\begin{equation}\label{eq:mobius}
x=\frac{z-i\beta}{z+i\beta},
\end{equation}
which maps \(z \in  \{0,\pm i\beta,\infty\}\) to \(x\in\{-1,0,1,\infty\}\). In the \(x\)-variable, the equation can be written in the form
\begin{equation}\label{eq:normalized-x}
W''(x)+\frac{2x}{x^2-1}W'(x)+Q(x)\,W(x)=0.
\end{equation}
The explicit expressions for \(Q(x)\), the Liouville-transformed potential, and the rational form of the \(z\)-equation are given in \autoref{app:heun}.
Finally, we remove the double poles by factoring off the Frobenius gauge.
\begin{equation}\label{eq:frobenius-gauge}
W(x)=x^{\rho_0}(x-1)^{\rho_1}(x+1)^{\rho_{-1}}\,y(x),
\end{equation}
where \(\rho_0,\rho_{\pm1}\) are chosen from the indicial equations so that the \((x-a)^{-2}\) terms cancel for \(a\in\{0,\pm1\}\). The resulting equation for \(y\) has four regular singular points at \(x\in\{-1,0,1,\infty\}\), so the mode equation is of general Heun type. We used this observation only as structural motivation; the actual spectral condition will be formulated later using the APS boundary conditions and the associated boundary determinant.

\subsection{APS boundary conditions on the warped cylinder}\label{subsec:aps}

In the previous subsection, we reduced the eigenvalue equation
$D \psi =  \lambda \psi$ to a second-order mode equation. To obtain a discrete spectrum of a finite cylinder, we need to specify an elliptic boundary condition. For this purpose, we impose APS boundary conditions, which require identifying the induced boundary operator on each boundary component.

The boundary is $\partial M = Y_0 \sqcup Y_T$ with $Y_0 = \{ 0 \} \times S^1$ and $Y_T  = \{ T \} \times S^1$.

\subsubsection{Normal vector and boundary Clifford multiplication}\label{subsec:boundary-clifford}
Let \(\{e_1,e_2\}\) be the orthonormal frame $e_1= \partial_t, e_2= \frac{1}{f(t)}\partial_\theta.$
Following the conventions of \cite{Hijazi}, we use the inward-pointing unit normal $N$ along each boundary
component. At $Y_0$ the inward unit normal is $N=+e_1$, whereas, at $Y_T$ the inward unit normal is
$N=-e_1$. We continue to write \(\gamma(\cdot)\) for the Hermitian gamma-matrix representation in the chosen orthonormal frame, so that
\begin{equation}
\gamma(e_1)=\gamma_1,\qquad \gamma(e_2)=\gamma_2.
\end{equation}
Consequently,
\begin{equation}
\gamma(N)=\gamma_1 \quad \text{at} \quad Y_0,
\qquad
\gamma(N)=-\gamma_1 \quad \text{at} \quad Y_T.
\end{equation}

For the warped metric, the bulk Dirac operator is not literally of pure product form near the boundary,
because its normal part contains the additional zeroth-order term
\begin{equation}
\frac{f'(t)}{2f(t)}.
\end{equation}
Consequently, we define the APS boundary operator intrinsically on each boundary component, following the hypersurface formalism of \cite{Hijazi}. To define a Hermitian Clifford action along the boundary, we set
\begin{equation}
c(X)=-\,i\,\gamma(N)\gamma(X),\qquad X\in \mathcal{T}(\partial M).
\end{equation}
Since \(X\perp N\), the Hermitian matrices \(\gamma(N)\) and \(\gamma(X)\) anticommute, so \(\gamma(N)\gamma(X)\) is skew-Hermitian; hence \(c(X)\) is Hermitian. Moreover, \(c(X)\) satisfies the Clifford relations for the induced boundary metric.
We equip \(L^2(Y_{t_0};S_{Y_{t_0}})\) with the natural \(L^2\)-inner product
\begin{equation}
\langle \phi,\psi\rangle_{L^2(Y_{t_0})}
=
\int_0^{2\pi}\langle \phi(\theta),\psi(\theta)\rangle_{\mathbb C^2}\,f(t_0)\,d\theta .
\end{equation}
With respect to this inner product, the tangential covariant derivative is skew-adjoint and the induced boundary Dirac operator is self-adjoint. In dimension two, the intrinsic boundary Dirac operator \(\tilde D_{Y_{t_0}}\) is related to the bulk Dirac operator by the hypersurface identity,
\begin{equation}
\tilde D_{Y_{t_0}}\psi
=
\frac12 H_{t_0}\,\psi-\gamma(N)D\psi-\nabla_N\psi,
\end{equation}
for spinors \(\psi\) obtained by restriction from the bulk. Here \(H_{t_0}\) denotes the mean curvature of the boundary component \(Y_{t_0}\), computed with respect to the chosen inward unit normal. In the present model, this intrinsic boundary operator is identified with the self-adjoint boundary operator \(B_{t_0}\) introduced below, and it is this operator that enters the APS projector.

 Therefore, for the unit tangent vector \(\mathcal{U}=e_2|_{t=t_0}\) on
\(Y_{t_0}\) one obtains
\begin{equation}    \label{eq:signs}
c(\mathcal{U})=
\begin{cases}
+\sigma_3,& t_0=0,\\
-\sigma_3,& t_0=T.
\end{cases}
\end{equation}
Crucially, the sign changes between the two ends because the inward normal changes sign.

\subsubsection{Self-adjoint boundary Dirac operator}\label{subsec:boundary-dirac}
With the line bundle \(E\) fixed, we suppress it from the notation here as well. The covariant derivative along the unit tangent at \(Y_{t_0}\) is
\begin{equation}
\nabla_\mathcal{U}=\frac{1}{f(t_0)}(\partial_\theta+iA).
\end{equation}
The operator \(\nabla_\mathcal{U}\) is skew-adjoint on \(L^2(Y_{t_0};S_{Y_{t_0}})\), so we introduce the Hermitian tangential operator
\begin{equation}
\mathcal{D}_{t_0} =-\,i\,\nabla_\mathcal{U}=\frac{1}{f(t_0)}\bigl(-i\partial_\theta + A\bigr),
\end{equation}
which has a real spectrum. We define the intrinsic boundary Dirac operator by
\begin{equation} \label{eq:30}
B_{t_0} =c(\mathcal{U})\,\mathcal{D}_{t_0}.
\end{equation} 
With the explicit expression for \(c(\mathcal{U})\) from \eqref{eq:signs}, this becomes
\begin{equation} \label{eq:31}
B_0=\frac{1}{f(0)}\,\sigma_3\bigl(-i\partial_\theta + A\bigr),\qquad
B_T=-\frac{1}{f(T)}\,\sigma_3\bigl(-i\partial_\theta + A\bigr).
\end{equation}
Each \(B_{t_0}\) is self-adjoint on \(L^2(Y_{t_0};S_{Y_{t_0}})\) with the natural domain
\(H^1(Y_{t_0};S_{Y_{t_0}})\).

We obtain the explicit mode matrices
\begin{equation} \label{eq:34}
B_0=\frac{m}{f(0)}\begin{pmatrix}1&0\\0&-1\end{pmatrix},\qquad
B_T=-\frac{m}{f(T)}\begin{pmatrix}1&0\\0&-1\end{pmatrix},
\end{equation} 
with \(m=k+A\). In the invertible case relevant to the APS discussion below, henceforth we assume \(m\neq 0\). Therefore, the boundary kernel is trivial, and both \(B_0\) and \(B_T\) are invertible.
The boundary eigenvalues are real and given by
\begin{equation}
\label{eq:mu}
\mu_{k,\pm}(t_0)=\pm \frac{m}{f(t_0)}.
\end{equation}
The assignment of the positive eigenvalue to the upper or lower spinor component is reversed between $t=0$ and $t=T$.

\subsubsection{APS boundary conditions}\label{subsec:aps-bc}
Let \(P_{>0}(B_{t_0})\) denote the \(L^2\)-orthogonal projection onto the direct sum of eigenspaces of \(B_{t_0}\) with positive eigenvalues. Since we are working in the case \(m\neq 0\), the boundary operators are invertible and \(P_{>0}(B_{t_0})=P_{\ge 0}(B_{t_0})\). The APS boundary condition is
\begin{equation}
P_{>0}(B_0)\bigl(\psi|_{t=0}\bigr)=0,\qquad
P_{>0}(B_T)\bigl(\psi|_{t=T}\bigr)=0.
\end{equation}
Since \(B_0\) and \(B_T\) are diagonal in \eqref{eq:34}, the APS boundary condition reduces to a single component constraint for each
mode \(m\neq 0\). Explicitly:
\begin{itemize}
\item If \(m>0\), then \(\mu > 0\) corresponds to the upper component when \(t=0\) and to the lower component when
\(t=T\), hence
\begin{equation}
u(0)=0,\qquad v(T)=0.
\end{equation}
\item If \(m<0\), then \(\mu > 0\) corresponds to the lower component when \(t=0\) and to the upper component when
\(t=T\), hence
\begin{equation}
v(0)=0,\qquad u(T)=0.
\end{equation}
\end{itemize}

With these boundary conditions imposed at both ends, each Fourier mode defines a Fredholm boundary equation for the bulk. The compatibility of the bulk solutions with the APS constraints determines the spectral condition and hence the spectrum.

\subsection{Bulk APS spectrum }\label{subsec:bulk-spectrum}
We introduce first-order differential operators
\begin{equation}
\mathcal{A}_{+} =\partial_t+\frac{f'}{2f}+\frac{m}{f},\qquad
\mathcal{A}_{-} =\partial_t+\frac{f'}{2f}-\frac{m}{f},
\end{equation}
so that the equations \eqref{eq:uv_system} become \(\mathcal{A}_+v=-i\lambda u\) and \(\mathcal{A}_-u=-i\lambda v\).

\subsubsection{Decoupling and reconstruction}\label{subsec:decouple-reconstruct}
Eliminating \(v\) yields the decoupled second-order equation for \(u\),
\begin{equation}\label{eq:u-decoupled}
\mathcal{A}_+\mathcal{A}_-u+\lambda^2 u=0.
\end{equation}
Similarly, eliminating \(u\) yields \(\mathcal{A}_-\mathcal{A}_+v+\lambda^2 v=0\). For \(\lambda\neq 0\), the second component can be
reconstructed from \(u\) using \eqref{eq:uv_system}:
\begin{equation}\label{eq:v-from-u}
v(t)=\frac{i}{\lambda}\Big(u'(t)+\Big(\frac{f'(t)}{2f(t)}-\frac{m}{f(t)}\Big)u(t)\Big),  \qquad \lambda \neq 0.
\end{equation}
Any nontrivial solution of \eqref{eq:u-decoupled} determines a unique spinor solution of
\eqref{eq:uv_system} after fixing an overall normalization.

\subsubsection{APS boundary conditions and boundary constraints}\label{subsec:aps-constraints}
The APS boundary condition imposes one
scalar constraint at each boundary component, and in the present basis, it takes the simple mode-by-mode form:
\begin{equation}
m>0:\quad u(0)=0,\ \ v(T)=0,
\qquad
m<0:\quad v(0)=0,\ \ u(T)=0.
\end{equation}

Alternatively, there exist linear functionals
\begin{equation}
b_0,b_T\in (\mathbb C^2)^\ast
\end{equation}
such that
\begin{equation}
b_0\bigl(\psi(0,\lambda)\bigr)=0,
\qquad
b_T\bigl(\psi(T,\lambda)\bigr)=0.
\end{equation}
With respect to the basis of the dual space \((\mathbb C^2)^\ast\), these functionals are
\begin{subequations}
\begin{align}
m>0:\quad & b_0=(1\ \ 0),\qquad b_T=(0\ \ 1), \\
m<0:\quad & b_0=(0\ \ 1),\qquad b_T=(1\ \ 0).
\end{align}
\end{subequations}

\begin{proposition}[Modewise APS operator]
Fix an allowed Fourier mode \(k\) and assume \(m=k+A\neq 0\). Let
\begin{equation}
D_k=
i\begin{pmatrix}
0 & \mathcal A_+\\
\mathcal A_- & 0
\end{pmatrix}
\end{equation}
act on $ \mathfrak H_k=L^2([0,T],f(t)\,dt;\C^2) $
with domain
\begin{equation}
\Dom(D_{k,\APS})=
\begin{cases}
\left\{\binom{u}{v}\in \mathfrak W_k : u(0)=0,\ v(T)=0\right\}, & m>0,\\
\left\{\binom{u}{v}\in \mathfrak W_k : v(0)=0,\ u(T)=0\right\}, & m<0.
\end{cases}
\end{equation}
Then \(D_{k,\APS}\) is a self-adjoint operator on \(\mathfrak H_k\) with compact resolvent. In particular, its spectrum is real, discrete, and of finite multiplicity.
\end{proposition}

\begin{proof}
Since \(f\in C^\infty([0,T])\) and \(f(t)>0\) on \([0,T]\), the weighted Sobolev space
\begin{equation}
\mathfrak W_k=H^1([0,T],f(t)\,dt;\C^2)
\end{equation}
coincides with \(H^1([0,T];\C^2)\) with the same norms. Let
\begin{equation}
\chi =\binom{u}{v},\qquad \nu=\binom{p}{q}\in \mathfrak W_k.
\end{equation}
A direct integration by parts in the weighted inner product on \(\mathfrak H_k\) gives
\begin{equation}
\langle D_k\chi,\nu\rangle_{\mathfrak H_k}
-
\langle \chi,D_k\nu\rangle_{\mathfrak H_k}
=
i f(0)\bigl(v(0)\overline{p(0)}-u(0)\overline{q(0)}\bigr)
-
i f(T)\bigl(v(T)\overline{p(T)}-u(T)\overline{q(T)}\bigr).
\end{equation}
If \(m>0\), the APS boundary conditions are
\begin{equation}
u(0)=0,\qquad v(T)=0,
\end{equation}
and the boundary form vanishes on \(\Dom(D_{k,\APS})\). If \(m<0\), the APS boundary conditions are $
v(0)=0,\qquad u(T)=0,$ and again the boundary form vanishes on \(\Dom(D_{k,\APS})\). Thus \(D_{k,\APS}\) is symmetric.

The associated boundary trace space is a two-dimensional subspace of $ \C^2\oplus \C^2,$
and in each case, it is maximal isotropic for the full boundary form above. Therefore, the symmetric operator is self-adjoint; see, for example, \cite{BoossWojciechowski,BaerBallmann2012}.

Finally, on the compact interval \([0,T]\), the graph norm of \(D_{k,\APS}\) is equivalent to the \(H^1\)-norm on the domain, and the embedding
\begin{equation}
H^1([0,T];\C^2)\hookrightarrow L^2([0,T];\C^2)
\end{equation}
is compact. Hence \(D_{k,\APS}\) has compact resolvent. The spectral conclusions follow.
\end{proof}

\subsubsection{Determinant formulation of the spectral condition}\label{subsec:determinant-spectrum}
Let \(\psi_1(t,\lambda)\) and \(\psi_2(t,\lambda)\) be two linearly independent solutions of
\eqref{eq:uv_system} on \([0,T]\). Then any solution can be written in the form
\begin{equation}
\psi(t,\lambda)=C_1\psi_1(t,\lambda)+C_2\psi_2(t,\lambda).
\end{equation}
Imposing the two APS constraints gives a homogeneous linear system for \((C_1,C_2)\), and a nontrivial solution exists if and only if the associated \(2\times 2\) boundary matrix is singular. Accordingly, we define
\begin{equation}\label{eq:Fk-def}
F_k(\lambda)=\det\begin{pmatrix}
b_0\psi_1(0,\lambda) & b_0\psi_2(0,\lambda)\\
b_T\psi_1(T,\lambda) & b_T\psi_2(T,\lambda)
\end{pmatrix}.
\end{equation}

\begin{proposition}[Determinant characterization of the modewise APS spectrum]
Fix an allowed Fourier mode \(k\) and assume \(m=k+A\neq 0\). Then
\begin{equation}
\lambda\in\Spec_{\mathrm{APS}}(k)
\quad\Longleftrightarrow\quad
F_k(\lambda)=0.
\end{equation}
Moreover, the zero set of \(F_k\) is independent of the chosen fundamental pair
\((\psi_1,\psi_2)\).
\end{proposition}

\begin{proof}
Any solution of \eqref{eq:uv_system} can be written uniquely as
\begin{equation}
\psi(t,\lambda)=C_1\psi_1(t,\lambda)+C_2\psi_2(t,\lambda).
\end{equation}
Imposing the APS boundary conditions at \(t=0\) and \(t=T\) gives a homogeneous linear system in \((C_1,C_2)\). A nontrivial solution exists if and only if the associated boundary matrix is singular, which is exactly the condition \(F_k(\lambda)=0\).

If another fundamental pair is obtained by
\begin{equation}
(\widetilde\psi_1,\widetilde\psi_2)=(\psi_1,\psi_2)M(\lambda),
\qquad
M(\lambda)\in\mathrm{GL}_2(\mathbb C),
\end{equation}
then
\begin{equation}
\widetilde F_k(\lambda)=\det(M(\lambda))\,F_k(\lambda),
\end{equation}
so the zero set is unchanged.
\end{proof}

In components, this reads:
\begin{subequations}
\begin{align}
m>0:\quad & F_k(\lambda)=u_1(0,\lambda)v_2(T,\lambda)-u_2(0,\lambda)v_1(T,\lambda), \\
m<0:\quad & F_k(\lambda)=v_1(0,\lambda)u_2(T,\lambda)-v_2(0,\lambda)u_1(T,\lambda).
\end{align}
\end{subequations}
%----------------------------
\subsubsection{Bulk Spectrum} \label{subsec:bulk_spectrum}
For a fixed Fourier mode $k$, let $m=k+A$ and consider the reduced first-order Dirac system \eqref{eq:uv_system}. 
To determine the APS spectrum numerically, we solve this system with the left APS boundary condition imposed at $t=0$ and define a scalar terminal residual at $t=T$.  When $m>0$ we impose the left APS condition $u(0)=0$ and define
\begin{equation}
F_k(\lambda)=v(T;\lambda),
\end{equation}
whereas when $m<0$ we impose the left APS condition $v(0)=0$ and define
\begin{equation}
F_k(\lambda)= u(T;\lambda).
\end{equation}
Thus, $F_k(\lambda)=0$ is exactly the remaining APS boundary condition, so the zeros of $F_k$ coincide with the APS eigenvalues in mode $k$. 
The plotted curve in \autoref{fig:F_k} is therefore not itself an eigenvalue branch, but a modewise characteristic function whose zero set is the spectrum.

\label{plots}

\begin{figure}[t]
  \centering
  \includegraphics[width=.85\linewidth]{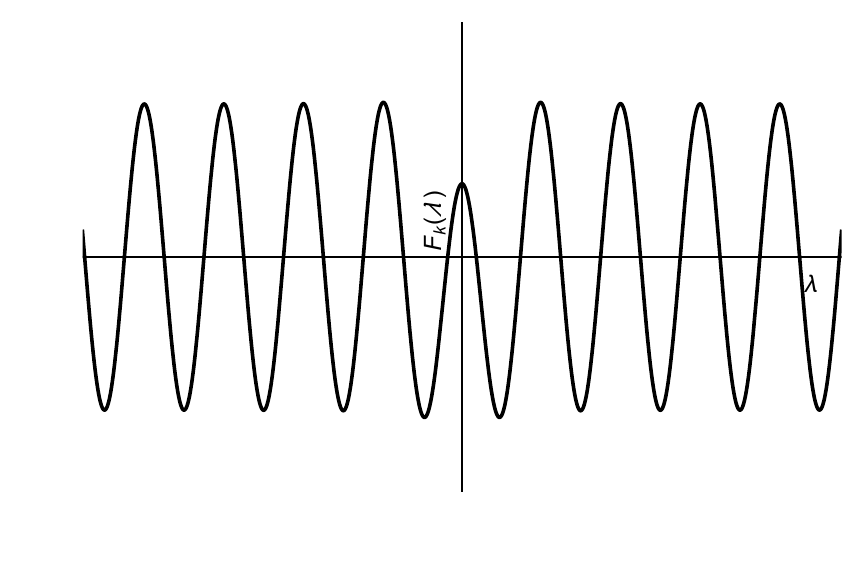}
 \caption{Modewise characteristic function $F_k(\lambda)$ for the APS boundary problem at $\alpha=1$, $A=0.3$, $k=1$, and $T=1.5$. The zeros of $F_k(\lambda)$ are the APS eigenvalues in this mode.}
  \label{fig:F_k}
\end{figure}

\begin{proposition}[Absence of zero modes in the invertible case]
\label{prop:no-zero-mode}
Fix an allowed Fourier mode $k$ and assume $m=k+A\neq 0$. Then $\lambda=0$ is
not an APS eigenvalue for the $k$-th mode.
\end{proposition}

\begin{proof}
For $\lambda=0$, the system \eqref{eq:uv_system} reduces to
\begin{equation}
\mathcal A_+v=0,\qquad \mathcal A_-u=0.
\end{equation}
Hence
\begin{equation}
v(t)=C_1 \exp\!\left(-\int_0^t\left(\frac{f'}{2f}+\frac{m}{f}\right)\,ds\right),
\qquad
u(t)=C_2 \exp\!\left(-\int_0^t\left(\frac{f'}{2f}-\frac{m}{f}\right)\,ds\right).
\end{equation}
If $m>0$, the APS conditions are $u(0)=0$ and $v(T)=0$, which force
$C_2=0$ and $C_1=0$. If $m<0$, the APS conditions are $v(0)=0$ and $u(T)=0$,
which again force $C_1=C_2=0$. Therefore, the only solution is the trivial
solution, so $\lambda=0$ is not an APS eigenvalue.
\end{proof}
%---------------------------

\section{The APS boundary correction: $\eta$-invariant and cancellation}
\label{sec:aps-boundary}
So far, we have considered the self-adjoint APS eigenvalue equation for the full Dirac operator $D$, whose modewise APS conditions constrain both components of the spinor.  Now we switch to the Fredholm APS problem for the chiral operator $D^+$. This shift is natural because the APS index theorem is formulated for the chiral Dirac operator, whereas the previous section concerned the self-adjoint spectral problem for the full Dirac operator.
\begin{equation}
D^{+}: \Gamma(S^{+}) \to \Gamma(S^{-}).
\end{equation}
For the chiral Dirac operator, the APS projector is imposed only on the $S^{+}$ boundary, while the complementary boundary condition appears in the adjoint problem for $D^{-}$. 

By \autoref{subsec:boundary-clifford}, the unit normal \(N\) is taken to be inward-pointing along each boundary component, and all boundary operators and APS projections below are understood in this convention. As in the hypersurface formalism above, the APS condition is defined from the spectral resolution of the intrinsic self-adjoint boundary Dirac operator on each boundary component. In the warped-cylinder model, these are the operators \(B_0\) and \(B_T\) in \eqref{eq:31}.

By \autoref{remark:spin_structure}, the allowed Fourier modes are
\begin{equation}
\mathcal K=\mathbb Z
\quad\text{in the periodic case,}\qquad
\mathcal K=\mathbb Z+\tfrac12
\quad\text{in the anti-periodic case.}
\end{equation}
%------------------------
\subsection{The boundary $\eta$-invariant and the reduced invariant $\xi$}
\label{subsec:boundary-eta}

Let $B$ be a self-adjoint first-order elliptic operator on a closed manifold.
For $\Re(s)\gg 1$, its $\eta$-function is defined by
\begin{equation}
\label{eq:eta-rewrite}
\eta(B,s)\;=\;\sum_{\mu\in \Spec(B)\setminus\{0\}} \sign(\mu)\,|\mu|^{-s}.
\end{equation}
It admits a meromorphic continuation to $s\in\C$ which is regular at $s=0$. We write \(\eta(B,0)\) for its value at \(s=0\).

In APS index theory, the boundary correction is expressed in terms of the reduced invariant
\begin{equation}
\label{eq:xi-def-header-rewrite}
\xi(B)\;=\;\frac{\eta(B,0)+h(B)}{2},
\qquad
h(B)=\dim\ker B.
\end{equation}
When $B$ is invertible, one has $h(B)=0$, hence,
\begin{equation}
\xi(B)=\frac{\eta(B,0)}{2}.
\end{equation}

We begin with the \(\eta\)-invariant computation for the circle operator and then return to the warped cylinder to show cancellation of the two APS boundary contributions when the gauge is constant.

\paragraph{Circle model.}
The intrinsic boundary Dirac operator is,
\begin{equation}
B_A e^{ik\theta}=(k+A)e^{ik\theta},
\qquad k\in\mathcal K.
\end{equation}
Assuming
\begin{equation}
k+A\neq 0 \qquad \text{for every } k\in\mathcal K,
\end{equation}
equivalently $0\notin\Spec(B_A)$.
We now justify the reduction to a unique parameter \(\rho\in(0,1)\). The open interval is forced by the invertibility assumption \(0\notin\Spec(B_A)\), since \(\rho=0\) or \(\rho=1\) would place \(0\) in the spectrum in the representation \eqref{eq:spec-rho-form}.

If $\mathcal K=\mathbb Z$, write
\begin{equation}
A=\ell+\rho,\qquad \ell\in\mathbb Z,\quad \rho\in(0,1).
\end{equation}
Then
\begin{equation}
\Spec(B_A)=\{n+A:n\in\mathbb Z\}
=\{n+\ell+\rho:n\in\mathbb Z\}
=\{m+\rho:m\in\mathbb Z\},
\end{equation}
after relabelling $m=n+\ell$.

If $\mathcal K=\mathbb Z+\tfrac12$, then
\begin{equation}
\Spec(B_A)=\Bigl\{n+\Bigl(A+\tfrac12\Bigr):n\in\mathbb Z\Bigr\}.
\end{equation}
Since $0\notin\Spec(B_A)$, the number $A+\tfrac12$ is not an integer. Hence, there is a unique decomposition
\begin{equation}
A+\tfrac12=\ell+\rho,\qquad \ell\in\mathbb Z,\quad \rho\in(0,1),
\end{equation}
and therefore
\begin{equation}
\Spec(B_A)=\{m+\rho:m\in\mathbb Z\},
\end{equation}
after relabelling the Fourier index.
Thus, in either spin structure, there is a unique $\rho\in(0,1)$ such that
\begin{equation}
\label{eq:spec-rho-form}
\Spec(B_A)=\{n+\rho:n\in\mathbb Z\}.
\end{equation}

Now we compute the $\eta$-function explicitly.
Using \eqref{eq:spec-rho-form}, for $\Re(s)>1$ one has absolute convergence and can split the positive and negative parts of the spectrum:
\begin{align}
\eta(B_A,s)
&=
\sum_{n+\rho>0}(n+\rho)^{-s}
-
\sum_{n+\rho<0}|n+\rho|^{-s} \notag\\
&=
\sum_{n=0}^{\infty}(n+\rho)^{-s}
-
\sum_{n=0}^{\infty}(n+1-\rho)^{-s}.
\label{eq:eta-rho-split}
\end{align}
Indeed, the negative eigenvalues are precisely
\begin{equation}
-(n+1-\rho),\qquad n=0,1,2,\dots.
\end{equation}
Therefore, for $\Re(s)>1$,
\begin{equation}
\label{eq:eta-hurwitz}
\eta(B_A,s)=\zeta(s,\rho)-\zeta(s,1-\rho),
\end{equation}
where $\zeta(s,a)$ denotes the Hurwitz zeta function. Each Hurwitz zeta function admits a meromorphic continuation to all $s\in\C$, with a simple pole at $s=1$ of residue $1$ \cite{Apostol}. Hence both $\zeta(s,\rho)$ and $\zeta(s,1-\rho)$ have the same pole at $s=1$. It follows that in the difference \eqref{eq:eta-hurwitz} the pole parts cancel. Hence \(\eta(B_A,s)\) admits a meromorphic continuation to all \(s\in\C\) and is regular at \(s=0\).

Evaluating at $s=0$ and using the standard formula
\begin{equation}
\zeta(0,a)=\frac12-a,
\qquad a\in(0,1),
\end{equation}
one obtains
\begin{equation}
\label{eq:eta-circle-rewrite}
\eta(B_A,0)
=
\Bigl(\frac12-\rho\Bigr)-\Bigl(\frac12-(1-\rho)\Bigr)
=
1-2\rho,
\end{equation}
and therefore
\begin{equation}
\label{eq:xi-circle-rewrite}
\xi(B_A)=\frac{1-2\rho}{2}.
\end{equation}

\paragraph{Relation to the actual endpoint operators.}
The operator \(B_A\) above is a normalized operator on the scalar boundary mode space over \(S^1\). After identifying the boundary components \(Y_0\) and \(Y_T\) with \(S^1\) via the coordinate \(\theta\), the actual positive-chirality endpoint operators are
\begin{equation}
B_0^+=-\frac1{f(0)}\,B_A,
\qquad
B_T^+=+\frac1{f(T)}\,B_A.
\end{equation}
Since \(f(0),f(T)>0\), positive rescaling does not change the reduced \(\eta\)-invariant:
\begin{equation}
\eta(cB,0)=\eta(B,0),\qquad \xi(cB)=\xi(B)\qquad (c>0),
\end{equation}
provided \(B\) is invertible. Thus, only the sign matters, and therefore
\begin{equation}
\xi(B_0^+)=-\xi(B_A),\qquad \xi(B_T^+)=\xi(B_A).
\end{equation}
This is the mechanism behind the cancellation of the two APS boundary contributions on the finite cylinder.
%------------------------
\subsection{Fredholm chiral APS problem and the APS index}
\label{subsec:fredholm-aps}

We now return to the chiral APS problem on the finite warped cylinder.
\begin{equation}
M=[0,T]\times S^1.
\end{equation}
Let $D^+:\Gamma(S^+)\to\Gamma(S^-)$ denote the positive-chirality part of the Dirac operator.
Because the metric and connection are $S^1$-invariant, both $D^+$ and the induced tangential boundary operators are diagonal in Fourier modes.

A spinor may be decomposed into Fourier modes as
\begin{equation}
\psi(t,\theta)=e^{ik\theta}\binom{u(t)}{v(t)},
\qquad k\in\mathcal K.
\end{equation}
In our matrix convention, $S^+$ is represented by the second component $v$, while $S^-$ is represented by the first component $u$.
The Dirac operator \eqref{eq:Diracop} takes the form,
\begin{equation}
D
=
i\begin{pmatrix}
0 & \mathcal{A}_+ \\
\mathcal{A}_- & 0
\end{pmatrix},
\qquad
\mathcal A_\pm
=
\partial_t+\frac{f'(t)}{2f(t)}\pm \frac{m}{f(t)},
\qquad
m=k+A.
\end{equation}
Thus,
\begin{equation}
D^+=i\mathcal{A}_+:\Gamma(S^+)\to\Gamma(S^-),
\qquad
D^-=i\mathcal{A}_-:\Gamma(S^-)\to\Gamma(S^+).
\end{equation}

Let $B_0^+$ and $B_T^+$ denote the tangential operators induced by $D^+$ at the boundary components $t=0$ and $t=T$, respectively. Similarly, let $B_0^-$ and $B_T^-$ denote the tangential operators induced by $D^-$. Under natural identification of both boundary circles with $S^1$, the change in boundary orientation together with the rescaling of the tangential metric implies that the endpoint operators differ by a sign and a positive constant factor.

The operators \(B_{t_0}\) introduced earlier act on the full boundary spinor bundle
\begin{equation}
S_{Y_{t_0}}=S^+_{Y_{t_0}}\oplus S^-_{Y_{t_0}}.
\end{equation}
In the chiral APS problem, we write \(B_{t_0}^\pm\) for the restrictions of \(B_{t_0}\) to the corresponding chiral summands.

\paragraph{Chiral decomposition of the boundary operator.}
Since
\begin{equation}
S_{Y_{t_0}}=S^-_{Y_{t_0}}\oplus S^+_{Y_{t_0}},
\end{equation}
and in our convention \(S^-\) is the upper component while \(S^+\) is the lower component, the full boundary operator decomposes diagonally as
\begin{equation}
B_{t_0}
=
\begin{pmatrix}
B_{t_0}^- & 0\\
0 & B_{t_0}^+
\end{pmatrix}.
\end{equation}
Thus \(B_{t_0}^\pm\) are simply the restrictions of \(B_{t_0}\) to the corresponding chiral summands. In particular, the modewise eigenvalues of \(B_{t_0}^\pm\) are read off directly from \eqref{eq:34}, and therefore
\begin{equation}
\label{eq:B0BTpm-rewrite}
B_0^+(k)=-\frac{m}{f(0)},
\qquad
B_T^+(k)=+\frac{m}{f(T)},
\qquad
B_0^-(k)=+\frac{m}{f(0)},
\qquad
B_T^-(k)=-\frac{m}{f(T)}.
\end{equation}
Here
\begin{equation}
U:L^2(Y_0;S_{Y_0})\longrightarrow L^2(Y_T;S_{Y_T})
\end{equation}
denotes the unitary induced by the identification of both boundary components with
\(S^1\) via the coordinate \(\theta\); in the chosen trivialization, it acts by
\begin{equation}
(U\phi)(\theta)=\phi(\theta).
\end{equation}
Hence, after identifying the two boundary Hilbert spaces by the unitary map $U$, the endpoint operators satisfy
\begin{equation}
\label{eq:BT-cB0-unitary}
B_T^\pm=-\,c\,U B_0^\pm U^{-1},
\qquad
c=\frac{f(0)}{f(T)}>0.
\end{equation}

\paragraph{Global APS domain.}

The APS chiral operator is,
\begin{equation}
D^+_{\APS}:\Dom(D^+_{\APS})\subset H^1(M;S^+)\longrightarrow L^2(M;S^-),
\end{equation}
with domain
\begin{equation}
\label{eq:global-aps-domain}
\Dom(D^+_{\APS})
=
\Bigl\{
\psi\in H^1(M;S^+):
P_{>0}(B_0^+)\bigl(\psi|_{t=0}\bigr)=0,\;
P_{>0}(B_T^+)\bigl(\psi|_{t=T}\bigr)=0
\Bigr\}.
\end{equation}
APS theory implies that the operator \(D^+_{\APS}\) is Fredholm.

For the positive-chirality \(k\)th Fourier mode, we write
\begin{equation}
\mathfrak H_k^+=L^2([0,T],f(t)\,dt;\C),\qquad
\mathfrak W_k^+=H^1([0,T],f(t)\,dt;\C).
\end{equation}

\medskip
\noindent\textbf{Modewise description.}
Due to $S^1$-invariance, $D^+_{\APS}$ decomposes into Fourier modes.
On the $k$-th mode, one obtains the scalar operator
\begin{equation}
D^+_{\APS,k}=i\mathcal{A}_+:
\Dom(D^+_{\APS,k})\subset \mathfrak W_k^+\to \mathfrak H_k^+,
\end{equation}
acting on the $v$-component.
For $m\neq 0$, the APS boundary condition becomes,
\begin{equation}
m>0:\quad v(T)=0 \ \text{(no condition at }0),
\qquad
m<0:\quad v(0)=0 \ \text{(no condition at }T).
\end{equation}
The adjoint APS conditions for $D^-$ impose the complementary endpoint conditions on $u$:
\begin{equation}
m>0:\quad u(0)=0,
\qquad
m<0:\quad u(T)=0.
\end{equation}

Thus, the chiral APS conditions are compatible with the two-component self-adjoint operator used in \autoref{subsec:aps-constraints}; modewise, they impose complementary endpoint conditions on $u$ and $v$.

%------------------------
\subsubsection{APS index formula and cancellation for constant gauge}
\label{subsubsec:eta-cancellation}

\paragraph{APS index formula.}
By the Atiyah-Patodi-Singer index theorem for Dirac operators on manifolds with boundary
(see \cite[Theorem~4.2]{APS1}; for background on conventions, see also \cite{LawsonMichelsohn,BGV}), the index is given by
\begin{equation}
\label{eq:APSindex}
\ind(D^+_{\APS})
=
\frac{i}{2\pi}\int_M F^{\nabla^E}
-\xi(B_0^+)-\xi(B_T^+).
\end{equation}
In constant gauge, the connection is flat, so \(F^{\nabla^E}=0\).

\paragraph{Geometric relation between the two endpoint operators.}
Equation \eqref{eq:BT-cB0-unitary} shows that, after identifying the two boundary circles, the tangential operators at $t=0$ and $t=T$ are related by a sign and a positive scale factor:
\begin{equation}
B_T^\pm=-\,c\,U B_0^\pm U^{-1},
\qquad
c=\frac{f(0)}{f(T)}>0.
\end{equation}
In particular, $B_T^+$ is unitarily equivalent to $-c\,B_0^+$.

\begin{lemma}[Sign/scale flip implies $\xi$-cancellation]
\label{lem:xi-sign-scale}
Let $B$ be a self-adjoint elliptic operator with discrete spectrum, and assume that $B$ is invertible.
If $c>0$, then
\begin{equation}
\eta(cB,0)=\eta(B,0),
\qquad
\eta(-B,0)=-\eta(B,0).
\end{equation}
Consequently, if
\begin{equation}
B_T=-\,c\,U B_0 U^{-1},
\qquad c>0,
\end{equation}
for some unitary $U$, then
\begin{equation}
\label{eq:xi-cancel-rewrite}
\xi(B_T)=-\xi(B_0),
\qquad\text{and hence}\qquad
\xi(B_0)+\xi(B_T)=0.
\end{equation}
\end{lemma}

\begin{proof}
When, $\Re(s)\gg 1$
\begin{equation}
\eta(cB,s)=c^{-s}\eta(B,s),
\end{equation}
so by meromorphic continuation $\eta(cB,0)=\eta(B,0)$.
Similarly,
\begin{equation}
\eta(-B,s)=-\eta(B,s),
\end{equation}
hence $\eta(-B,0)=-\eta(B,0)$.
Unitary conjugation does not change the spectrum, so
\begin{equation}
\eta(U B U^{-1},0)=\eta(B,0).
\end{equation}
If $B$ is invertible, then $h(B)=0$ and therefore $\xi(B)=\eta(B,0)/2$.
Applying these facts to $B_T=-c\,U B_0 U^{-1}$ yields \eqref{eq:xi-cancel-rewrite}.
\end{proof}

\begin{proposition}[Vanishing of the APS boundary correction for constant gauge]
\label{prop:aps-cancellation-constant-gauge}
Assume that $k+A\neq 0$ for every $k\in\mathcal K$. Then the endpoint operators $B_0^+$ and $B_T^+$ are invertible, and
\begin{equation}
\xi(B_0^+)+\xi(B_T^+)=0.
\end{equation}
Consequently,
\begin{equation}
\ind(D^+_{\APS})=0.
\end{equation}
\end{proposition}

\begin{proof}
By assumption, one has $m=k+A\neq 0$ for every $k\in\mathcal K$, so the modewise eigenvalues in \eqref{eq:B0BTpm-rewrite} never vanish. Thus, both $B_0^+$ and $B_T^+$ are invertible.
By \eqref{eq:BT-cB0-unitary} and Lemma~\ref{lem:xi-sign-scale},
\begin{equation}
\xi(B_T^+)=-\xi(B_0^+).
\end{equation}
Hence, the APS boundary correction cancels:
\begin{equation}
\xi(B_0^+)+\xi(B_T^+)=0.
\end{equation}
Since $F^{\nabla^E}=0$, the index formula \eqref{eq:APSindex} reduces to
\begin{equation}
\ind(D^+_{\APS})=0.
\end{equation}
\end{proof}

\begin{remark}[Wall mode and parameter variation]
\label{rem:wall-mode}
When $A$ is constant, no wall mode occurs precisely when $k+A\neq 0$ for every $k\in\mathcal K$.
The wall case $m=0$ becomes relevant when one varies $A$ (or, more generally, a boundary family) so that $A$ crosses the lattice of allowed shifts; depending on the spin structure, this lattice is $\Z$ or $\Z+\tfrac12$.
In that situation, $0\in\Spec(B)$ and the APS projection is no longer canonical on $\ker B$.
We return to this issue in the discussion of spectral flow.
\end{remark}

\begin{remark}[No APS boundary condition at infinity]
\label{rem:no-eta-at-infinity}
The cancellation above is a finite-cylinder statement: it uses the fact that $M$ has two genuine boundary components, and hence two APS boundary operators $B_0$ and $B_T$.
If one passes formally to the half-infinite cylinder $[0,\infty)\times S^1$ without imposing additional asymptotic conditions, there is no second APS boundary contribution and therefore no literal ``$\eta$-invariant at infinity'' canceling $\eta(B_0,0)$.

Consequently, throughout this paper, the phrase ``$T=\infty$'' is interpreted only as the large-$T$ limit of the finite APS problem on $[0,T]\times S^1$, not as an APS boundary condition imposed at infinity.
In particular, for a constant $A$ under the invertible assumption, the APS index vanishes for every finite $T$, and therefore also in the large-$T$ limit.
\end{remark}

%============================================================
% Section: Maslov index and bulk spectral flow 
% Conventions: inward normal flips; c(u)=+sigma3 at t=0, -sigma3 at t=T
%==========================================================
\section{Regularized boundary families, Maslov index, and zero-mode crossings}\label{sec:maslov}

Our motivation for the analysis in this section is to isolate the crossing mechanism associated with the boundary-zero set
\begin{equation}
k+A(s)=0.
\end{equation}
In the two-boundary APS setup studied in Sections~\ref{sec:2ddirac}-\ref{sec:aps-boundary}, the endpoint reduced $\eta$-invariant contributions cancel in the constant gauge invertible case, so the APS index does not directly provide a continuous framework for analyzing these boundary-zero crossings. For this reason, we introduce a regularized family of self-adjoint boundary conditions whose parameter dependence is continuous. This makes the standard spectral-flow/Maslov formalism accessible and enables direct analysis of the zero-mode crossing set.

\paragraph{Assumption.}
Throughout this section, we fix a Fourier mode \(k\in\mathcal K\) and a smoothing parameter \(\delta>0\). Unless explicitly stated otherwise, all operators and boundary conditions below refer to the regularized family \(D^{(\delta)}_{s,k}\). We assume endpoint invertibility at \(s=s_1,s_2\). Whenever a signed local crossing count is used, we assume the relevant crossings are isolated and simple.

\begin{theorem}[Zero-mode criterion for the regularized family]
\label{thm:section5-main}
Fix a mode \(k\in\mathcal K\) and assume
\begin{equation}
\delta \neq \frac{2}{\ell(T)},
\qquad
\ell(T)=\int_0^T \frac{d\tau}{f(\tau)}.
\end{equation}
Let \(\{D^{(\delta)}_{s,k}\}_{s\in[s_1,s_2]}\) denote the regularized \(k\)-mode operator family defined in \autoref{subsec:regularized_realizations_maslov}. Then
\begin{equation}
0\in\Spec\bigl(D^{(\delta)}_{s,k}\bigr)
\quad\Longleftrightarrow\quad
k+A(s)=0.
\end{equation}
Thus, zero modes of the regularized family occur exactly at the boundary zeros. If moreover
\(A'(s_*)\neq 0\) at such a point \(s_*\), then \(s_*\) is an isolated regular crossing for the regularized family, and the local spectral-flow/Maslov crossing formalism applies.
\end{theorem}

\begin{corollary}[Global zero-mode criterion]
\label{cor:global-zero-mode-criterion}
Under the same assumptions, the global regularized family has a zero mode at parameter \(s\) if and only if there exists \(k\in\mathcal K\) such that
\begin{equation}
k+A(s)=0.
\end{equation}
\end{corollary}

\paragraph{Steps to prove Theorem~\ref{thm:section5-main}}
First, we identify the boundary-zero set and define the regularized endpoint coefficients. Then, we rewrite the mode equation in a real first-order form and describe the associated boundary symplectic space and Lagrangians. Next, we show that the resulting regularized boundary problem defines a continuous self-adjoint Fredholm family to which the standard spectral-flow/Maslov correspondence applies. Finally, we compute the zero-mode matching condition explicitly and deduce the zero-mode criterion and the regularity of transverse crossings.

\subsection{Boundary zeros and regularized endpoint conditions}\label{subsec:boundary_zeros_regularized_bc}

The regularization replaces the discontinuous sign choice in the APS boundary condition by a continuous parameter \(m(s,k)=k+A(s)\). For \(m(s,k)\neq 0\), the regularized endpoint family converges as \(\delta\to 0\) to the corresponding APS choice, while remaining continuous through \(m(s,k)=0\).

Let \(A=A(s)\) be a smooth one-parameter family. The endpoint boundary operators are
\begin{equation}\label{eq:Bs-def}
B_s(0)=\frac{1}{f(0)}\,\sigma_3\bigl(-i\partial_\theta+A(s)\bigr),
\qquad
B_s(T)=-\frac{1}{f(T)}\,\sigma_3\bigl(-i\partial_\theta+A(s)\bigr),
\qquad s\in[s_1,s_2].
\end{equation}
Since these operators are \(S^1\)-invariant, they are diagonal in Fourier modes. Writing
\begin{equation}
\mathcal K=
\begin{cases}
\mathbb Z, & \text{periodic spin structure},\\
\mathbb Z+\tfrac12, & \text{anti-periodic spin structure},
\end{cases}
\qquad
m(s,k)=k+A(s),
\qquad k\in\mathcal K,
\end{equation}
their restriction to the \(k\)-th Fourier mode is
\begin{equation}\label{eq:Bs-eigs}
B_s(0)\big|_k
=
\frac{m(s,k)}{f(0)}
\begin{pmatrix}
1&0\\0&-1
\end{pmatrix},
\qquad
B_s(T)\big|_k
=
-\frac{m(s,k)}{f(T)}
\begin{pmatrix}
1&0\\0&-1
\end{pmatrix}.
\end{equation}
Hence, the boundary eigenvalues in mode \(k\) are
\begin{equation}
\mu_{k,\pm}(s;t_0)=\pm \frac{m(s,k)}{f(t_0)},
\qquad t_0\in\{0,T\}.
\end{equation}

A \emph{boundary-zero} for mode \(k\) is a parameter \(s_*\) such that
\begin{equation}\label{eq:mass-root}
m(s_*,k)=k+A(s_*)=0.
\end{equation}
We assume endpoint invertibility,
\begin{equation}
0\notin\Spec(B_{s_1}(t_0)),
\qquad
0\notin\Spec(B_{s_2}(t_0)),
\qquad t_0\in\{0,T\},
\end{equation}
so no boundary-zero occurs at \(s=s_1\) or \(s=s_2\).

A boundary-zero \(s_*\) is called \emph{transverse} if
\begin{equation}
\partial_s m(s_*,k)=A'(s_*)\neq 0.
\end{equation}
In the present scalar modewise setting, every transverse boundary-zero is simple and hence regular for the local Maslov crossing form. If \(A'(s_*)=0\), then the boundary-zero is nontransverse; it may be a genuine sign-changing crossing or merely a touching point. A sufficient local criterion for touching is
\begin{equation}
m(s_*,k)=0,\qquad \partial_s m(s_*,k)=0,\qquad \partial_s^2 m(s_*,k)\neq 0,
\end{equation}
equivalently,
\begin{equation}
k+A(s_*)=0,\qquad A'(s_*)=0,\qquad A''(s_*)\neq 0.
\end{equation}

Because \(A(s)\) is bounded on the compact interval \([s_1,s_2]\), only finitely many modes satisfy \(k+A(s)=0\) for some \(s\in[s_1,s_2]\).

Fix \(\delta>0\). On the scalar boundary mode space over \(S^1\), set
\begin{equation}
M_s=-i\partial_\theta+A(s),
\end{equation}
and defined by functional calculus
\begin{equation}\label{eq:alpha-global}
\alpha_0(s)=\tanh\!\Big(\frac{M_s}{\delta}\Big),
\qquad
\alpha_T(s)=-\alpha_0(s).
\end{equation}
On the \(k\)-th mode,
\begin{equation}\label{eq:alpha_def_clean}
\alpha_0(s,k)=\tanh\!\Big(\frac{m(s,k)}{\delta}\Big),
\qquad
\alpha_T(s,k)=-\alpha_0(s,k).
\end{equation}

\begin{remark}[Why this regularization is chosen]
The choice
\begin{equation}
\alpha_0(s,k)=\tanh\!\Big(\frac{m(s,k)}{\delta}\Big),
\qquad
\alpha_T(s,k)=-\alpha_0(s,k),
\end{equation}
is a choice rather than a canonical one. It gives a smooth interpolation between the two APS endpoint choices as the sign of \(m(s,k)\) changes, and the identity
\begin{equation}
\frac{1+\tanh z}{1-\tanh z}=e^{2z}
\end{equation}
makes the zero-mode matching equation explicitly solvable.
\end{remark}

\subsection{Real formulation, boundary form, and Lagrangians}\label{subsec:real_formulation_symplectic_space}

We now pass to a real first-order formulation adapted to the boundary symplectic form. Write the original mode spinor as
\begin{equation}
\psi=\binom{u}{v}.
\end{equation}
Let
\begin{equation}
U=\operatorname{diag}(1,i),
\end{equation}
and define
\begin{equation}
\mathcal U:L^2([0,T],f(t)\,dt;\mathbb C^2)\longrightarrow L^2([0,T],dt;\mathbb C^2),
\qquad
(\mathcal U\psi)(t)=f(t)^{1/2}U^\ast \psi(t).
\end{equation}
If
\begin{equation}
\phi=\mathcal U\psi=\binom{\tilde u}{w},
\qquad
\tilde u=f(t)^{1/2}u,
\qquad
w=-\,i\,f(t)^{1/2}v,
\end{equation}
then \(\mathcal U\) is unitary. We relabel \(\tilde u\) as \(u\), and write
\begin{equation}
\phi=\binom{u}{w}.
\end{equation}

Conjugating by \(\mathcal U\) transforms the mode differential expression into
\begin{equation}\label{eq:Duw-operator}
\mathcal D_{s,k}
=
\begin{pmatrix}
0 & -\partial_t-\dfrac{m(s,k)}{f(t)} \\
\partial_t-\dfrac{m(s,k)}{f(t)} & 0
\end{pmatrix}.
\end{equation}
Hence,
\begin{equation}\label{eq:uw-system}
\mathcal D_{s,k}\phi=\lambda\phi
\quad\Longleftrightarrow\quad
\begin{aligned}
u'(t) &= \frac{m(s,k)}{f(t)}\,u(t)+\lambda\,w(t),\\
w'(t) &= -\frac{m(s,k)}{f(t)}\,w(t)-\lambda\,u(t),
\end{aligned}
\qquad t\in[0,T].
\end{equation}
The coefficients are real.

In these variables, the regularized endpoint conditions are
\begin{equation}\label{eq:bc_tanh_clean}
\begin{aligned}
(1+\alpha_0(s,k))\,u(0)+(1-\alpha_0(s,k))\,w(0)&=0,\\
(1+\alpha_T(s,k))\,u(T)+(1-\alpha_T(s,k))\,w(T)&=0.
\end{aligned}
\end{equation}

\begin{lemma}[Real and complex zero modes]
\label{lem:real-complex-kernel}
Fix \(s\), \(k\), and \(\delta>0\), and let \(\mathcal D^{(\delta)}_{s,k}\) denote the operator associated with \eqref{eq:Duw-operator} and the endpoint conditions \eqref{eq:bc_tanh_clean}. Then, complex conjugation preserves \(\Dom(\mathcal D^{(\delta)}_{s,k})\) and commutes with \(\mathcal D^{(\delta)}_{s,k}\). Consequently,
\begin{equation}
\ker_{\C}\bigl(\mathcal D^{(\delta)}_{s,k}\bigr)
=
\ker_{\R}\bigl(\mathcal D^{(\delta)}_{s,k}\bigr)\otimes_{\R}\C.
\end{equation}
In particular, if the real kernel is one-dimensional, then the complex kernel is one-dimensional as a complex vector space.
\end{lemma}

\begin{proof}
The coefficients of \eqref{eq:Duw-operator} and the endpoint relations \eqref{eq:bc_tanh_clean} are real. Therefore, complex conjugation preserves both the domain and the action of the operator. If \(x\) is a complex zero mode, then its real and imaginary parts are real zero modes satisfying the same endpoint conditions, which gives the stated complexification identity.
\end{proof}

Let
\begin{equation}
x(t)=\begin{pmatrix}u(t)\\ w(t)\end{pmatrix},
\qquad
y(t)=\begin{pmatrix}p(t)\\ q(t)\end{pmatrix},
\end{equation}
with \(x,y\in H^1([0,T];\mathbb R^2)\). Integration by parts gives
\begin{equation}\label{eq:green-identity-uw}
\langle \mathcal D_{s,k}x,y\rangle_{L^2}
-
\langle x,\mathcal D_{s,k}y\rangle_{L^2}
=
\bigl[u(t)q(t)-w(t)p(t)\bigr]_{t=0}^{t=T}.
\end{equation}
If
\begin{equation}
J'=
\begin{pmatrix}
0&-1\\
1&0
\end{pmatrix},
\end{equation}
define
\begin{equation}
\omega(\xi,\eta)=\langle J'\xi,\eta\rangle_{\mathbb R^2},
\qquad
\xi,\eta\in\mathbb R^2.
\end{equation}
On the endpoint trace space
\begin{equation}
\mathcal H_k=\mathbb R^2\oplus\mathbb R^2,
\qquad
(x_0,x_T)=(x(0),x(T)),
\end{equation}
set
\begin{equation}\label{eq:Omega-def-green}
\Omega\big((x_0,x_T),(y_0,y_T)\big)
=
\omega(x_T,y_T)-\omega(x_0,y_0).
\end{equation}
Then \eqref{eq:green-identity-uw} becomes
\begin{equation}\label{eq:green-identity-Omega}
\langle \mathcal D_{s,k}x,y\rangle_{L^2}
-
\langle x,\mathcal D_{s,k}y\rangle_{L^2}
=
\Omega\big((x(0),x(T)),(y(0),y(T))\big).
\end{equation}
Thus, separated self-adjoint boundary conditions are encoded by Lagrangian subspaces of \((\mathcal H_k,\Omega)\).

We now define the two Lagrangians relevant for the zero-mode crossing problem. First,
\begin{subequations}
\begin{align}
\Lambda_k^0(s)
&=
\Bigl\{x=(u,w)^{\mathsf T}\in\mathbb R^2:\ (1+\alpha_0(s,k))u+(1-\alpha_0(s,k))w=0\Bigr\},\\
\Lambda_k^T(s)
&=
\Bigl\{x=(u,w)^{\mathsf T}\in\mathbb R^2:\ (1+\alpha_T(s,k))u+(1-\alpha_T(s,k))w=0\Bigr\},
\end{align}
\end{subequations}
and
\begin{equation}\label{eq:Lambda_bc_clean}
\Lambda_{\mathrm{bc},k}(s)=\Lambda_k^0(s)\oplus\Lambda_k^T(s)\subset\mathcal H_k.
\end{equation}

\begin{lemma}
\label{lem:bc-lagrangian}
For each fixed \(s\) and \(k\), the subspace \(\Lambda_{\mathrm{bc},k}(s)\subset\mathcal H_k\) is Lagrangian.
\end{lemma}

\begin{proof}
Each of \(\Lambda_k^0(s)\) and \(\Lambda_k^T(s)\) is a one-dimensional isotropic subspace of \((\mathbb R^2,\omega)\). Their direct sum is therefore a two-dimensional maximal isotropic subspace of \((\mathcal H_k,\Omega)\).
\end{proof}

Second, at \(\lambda=0\) the system \eqref{eq:uw-system} decouples:
\begin{equation}
u'=\frac{m(s,k)}{f(t)}u,
\qquad
w'=-\frac{m(s,k)}{f(t)}w.
\end{equation}
Let
\begin{equation}\label{eq:ell_clean}
\ell(t)=\int_0^t \frac{d\tau}{f(\tau)}.
\end{equation}
Then the transfer matrix is
\begin{equation}\label{eq:transfer_clean}
x(T)=\Phi_{s,k,0}(T)\,x(0),
\qquad
\Phi_{s,k,0}(T)=
\begin{pmatrix}
e^{\,m(s,k)\,\ell(T)}&0\\
0&e^{-\,m(s,k)\,\ell(T)}
\end{pmatrix}.
\end{equation}
Since \(\Phi_{s,k,0}(T)\) is symplectic with respect to \(\omega\), its graph
\begin{equation}\label{eq:Lambda_cauchy_clean}
\Lambda_{s,k}(0)
=
\operatorname{graph}\bigl(\Phi_{s,k,0}(T)\bigr)
=
\Bigl\{(x_0,x_T)\in\mathcal H_k:\ x_T=\Phi_{s,k,0}(T)x_0\Bigr\}
\end{equation}
is Lagrangian.

\begin{lemma}
\label{lem:cauchy-lagrangian}
For each fixed \(s\) and \(k\), the subspace \(\Lambda_{s,k}(0)\subset\mathcal H_k\) is Lagrangian.
\end{lemma}

\begin{proof}
The graph of a symplectic map is Lagrangian in the product symplectic space \((\mathcal H_k,\Omega)\).
\end{proof}

\begin{lemma}[Limit of the regularized endpoint family]
\label{lem:sharp-limit-regularized}
Fix a mode \(k\in\mathcal K\) and a parameter value \(s\) such that
\begin{equation}
m(s,k)=k+A(s)\neq 0.
\end{equation}
Then, as \(\delta\downarrow 0\), the endpoint lines \(\Lambda_k^0(s)\) and \(\Lambda_k^T(s)\) converge to the corresponding APS endpoint lines:
\begin{equation}
m(s,k)>0
\quad\Longrightarrow\quad
\Lambda_k^0(s)\to\{u=0\},\qquad
\Lambda_k^T(s)\to\{w=0\},
\end{equation}
and
\begin{equation}
m(s,k)<0
\quad\Longrightarrow\quad
\Lambda_k^0(s)\to\{w=0\},\qquad
\Lambda_k^T(s)\to\{u=0\}.
\end{equation}
\end{lemma}

\begin{proof}
If \(m(s,k)\neq 0\), then
\begin{equation}
\alpha_0(s,k)=\tanh\!\Big(\frac{m(s,k)}{\delta}\Big)\longrightarrow \sign(m(s,k))
\qquad (\delta\downarrow 0),
\end{equation}
and \(\alpha_T(s,k)=-\alpha_0(s,k)\). Substituting these limits into \eqref{eq:bc_tanh_clean} yields the claim.
\end{proof}

\subsection{The regularized boundary and the Maslov formalism}\label{subsec:regularized_realizations_maslov}

For fixed \(s\), \(k\), and \(\delta>0\), let \(\mathcal D^{(\delta)}_{s,k}\) denote the operator associated with \eqref{eq:Duw-operator} on \(L^2([0,T],dt;\mathbb C^2)\) and the regularized endpoint conditions
\begin{equation}
\Dom(\mathcal D^{(\delta)}_{s,k})
=
\left\{
x=\binom{u}{w}\in H^1([0,T];\mathbb C^2):
\begin{array}{l}
(1+\alpha_0(s,k))u(0)+(1-\alpha_0(s,k))w(0)=0,\\
(1+\alpha_T(s,k))u(T)+(1-\alpha_T(s,k))w(T)=0
\end{array}
\right\}.
\end{equation}
Define the corresponding operator in the original mode variables by
\begin{equation}
D^{(\delta)}_{s,k} =\mathcal U^{-1}\mathcal D^{(\delta)}_{s,k}\mathcal U.
\end{equation}

At the global level,
\begin{equation}\label{eq:global-regularized-decomposition}
D^{(\delta)}_{s}
=
\widehat{\bigoplus}_{k\in\mathcal K} D^{(\delta)}_{s,k},
\end{equation}
but all arguments below are modewise.

\begin{proposition}[Fixed-\texorpdfstring{$s$}{s} regularized operator]
\label{prop:fixed-s-selfadjoint}
Fix \(k\in\mathcal K\), \(\delta>0\), and \(s\in[s_1,s_2]\). Then \(\mathcal D^{(\delta)}_{s,k}\) is a self-adjoint operator with separated boundary conditions associated with the mode differential expression \eqref{eq:Duw-operator}.
\end{proposition}

\begin{proof}
By \eqref{eq:green-identity-Omega}, the operator with separated boundary conditions is symmetric exactly when its endpoint trace space is isotropic in \((\mathcal H_k,\Omega)\). For fixed \(s\), the regularized endpoint conditions define the trace space
\begin{equation}
\Lambda_{\mathrm{bc},k}(s)=\Lambda_k^0(s)\oplus\Lambda_k^T(s),
\end{equation}
which is Lagrangian by Lemma~\ref{lem:bc-lagrangian}. Hence, the operator is symmetric and maximally isotropic, therefore self-adjoint.
\end{proof}

\begin{corollary}[Compact resolvent]
\label{cor:compact-resolvent-regularized}
For each fixed \(s\in[s_1,s_2]\), the operators \(\mathcal D^{(\delta)}_{s,k}\) and \(D^{(\delta)}_{s,k}\) have compact resolvent. In particular, their spectra are real, discrete, and consist of eigenvalues of finite multiplicity.
\end{corollary}

\begin{proof}
The operator is first order on the compact interval \([0,T]\) with \(H^1\)-domain and separated homogeneous boundary conditions. Hence the graph norm is equivalent to the \(H^1\)-norm on the domain, and the embedding
\begin{equation}
H^1([0,T];\mathbb C^2)\hookrightarrow L^2([0,T];\mathbb C^2)
\end{equation}
is compact. Unitary equivalence preserves compactness of the resolvent and the spectrum.
\end{proof}

\begin{lemma}[Continuity of the regularized mode family]
\label{lem:gap-continuity-regularized}
Fix \(k\in\mathcal K\) and \(\delta>0\). Then the maps
\begin{equation}
s\longmapsto \Lambda_{\mathrm{bc},k}(s),
\qquad
s\longmapsto \Lambda_{s,k}(0)
\end{equation}
are continuous in the Lagrangian Grassmannian of \((\mathcal H_k,\Omega)\), and the corresponding self-adjoint family
\begin{equation}
s\longmapsto \mathcal D^{(\delta)}_{s,k}
\end{equation}
it is continuous in the gap topology. Equivalently, the family
\begin{equation}
s\longmapsto D^{(\delta)}_{s,k}
\end{equation}
is gap-continuous.
\end{lemma}

\begin{proof}
The coefficients of \eqref{eq:Duw-operator} depend continuously on \(s\) through \(m(s,k)=k+A(s)\). The boundary coefficients
\begin{equation}
\alpha_0(s,k)=\tanh\!\Big(\frac{m(s,k)}{\delta}\Big),
\qquad
\alpha_T(s,k)=-\alpha_0(s,k)
\end{equation}
therefore vary continuously in \(s\), so \(\Lambda_{\mathrm{bc},k}(s)\) varies continuously. The same is true for the zero-mode propagator \(\Phi_{s,k,0}(T)\) and hence for its graph \(\Lambda_{s,k}(0)\). Standard results for separated self-adjoint first-order systems on a compact interval then yield gap continuity; see, for example, \cite{Nicolaescu1995,BaerBallmann2012}.
\end{proof}

\begin{corollary}[Modewise spectral flow equals Maslov index]
\label{cor:sf-equals-maslov}
Fix \(k\in\mathcal K\) and \(\delta>0\). Then
\begin{equation}
\operatorname{SF}\bigl(\{D^{(\delta)}_{s,k}\}_{s\in[s_1,s_2]};0\bigr)
=
\mu_{\mathrm{Maslov}}\!\bigl(\Lambda_{s,k}(0),\,\Lambda_{\mathrm{bc},k}(s)\bigr).
\end{equation}
\end{corollary}

\begin{proof}
By Proposition~\ref{prop:fixed-s-selfadjoint}, Corollary~\ref{cor:compact-resolvent-regularized}, and Lemma~\ref{lem:gap-continuity-regularized}, the family \(\{D^{(\delta)}_{s,k}\}_{s\in[s_1,s_2]}\) is a continuous path of self-adjoint Fredholm operators in the standard separated-boundary setting. The usual spectral-flow/Maslov correspondence therefore applies; see, for example, \cite[Theorem~3.14]{Nicolaescu1995}.
\end{proof}

\begin{remark}[Scope]
\label{rem:scope-maslov-regularized-only}
The spectral-flow/Maslov correspondence is used here only for the continuous regularized family \(D^{(\delta)}_{s,k}\). It is not applied to the APS family across boundary zeros \(k+A(s)=0\), since the APS projector is discontinuous there.
\end{remark}

\subsection{Explicit zero-mode criterion and proof of the main theorem}\label{subsec:matching_clean}

We now compute the zero-mode crossing condition. By the previous subsections,
\begin{equation}
0\in\Spec\bigl(D^{(\delta)}_{s,k}\bigr)
\quad\Longleftrightarrow\quad
\Lambda_{s,k}(0)\cap\Lambda_{\mathrm{bc},k}(s)\neq\{0\}.
\end{equation}
Equivalently, there exists a nonzero solution of the zero-mode system satisfying the regularized endpoint conditions \eqref{eq:bc_tanh_clean}.

We repeatedly use the non-degenerate condition
\begin{equation}\label{eq:non-degenerate-delta}
\delta \neq \frac{2}{\ell(T)},
\qquad
\ell(T)=\int_0^T \frac{d\tau}{f(\tau)}.
\end{equation}

Choose a nonzero spanning vector for the left endpoint line \(\Lambda_k^0(s)\):
\begin{equation}
v_0(s,k)=\binom{1-\alpha_0(s,k)}{-(1+\alpha_0(s,k))}\in\Lambda_k^0(s).
\end{equation}
Let
\begin{equation}
x(t;s,\lambda,k)=\binom{u(t;s,\lambda,k)}{w(t;s,\lambda,k)}
\end{equation}
denote the unique solution of \eqref{eq:uw-system} with initial condition
\begin{equation}
x(0;s,\lambda,k)=v_0(s,k).
\end{equation}
Define the scalar matching function
\begin{equation}\label{eq:F_matching_clean}
F(s,\lambda,k)=
(1+\alpha_T(s,k))\,u(T;s,\lambda,k)
+
(1-\alpha_T(s,k))\,w(T;s,\lambda,k).
\end{equation}
Then \(F(s,\lambda,k)=0\) if and only if the corresponding solution satisfies both regularized endpoint conditions.

\begin{proposition}[Explicit zero-mode criterion for the regularized family]
\label{prop:explicit-zero-mode-criterion}
Fix a mode \(k\) and write \(m(s,k)=k+A(s)\). For \(\lambda=0\), the matching function \eqref{eq:F_matching_clean} is
\begin{equation}\label{eq:F-zero-explicit}
F(s,0,k)
=
\Bigl(1-\tanh\!\frac{m(s,k)}{\delta}\Bigr)^2 e^{\,m(s,k)\ell(T)}
-
\Bigl(1+\tanh\!\frac{m(s,k)}{\delta}\Bigr)^2 e^{-\,m(s,k)\ell(T)}.
\end{equation}
Equivalently,
\begin{equation}\label{eq:F-zero-factor}
F(s,0,k)=0
\quad\Longleftrightarrow\quad
m(s,k)\Bigl(\ell(T)-\frac{2}{\delta}\Bigr)=0.
\end{equation}
In particular, under \eqref{eq:non-degenerate-delta},
\begin{equation}
F(s,0,k)=0
\quad\Longleftrightarrow\quad
m(s,k)=0.
\end{equation}
Hence, zero modes of the regularized family occur exactly at the boundary zeros.
\end{proposition}

\begin{proof}
At \(\lambda=0\), the system decouples:
\begin{equation}
u(t)=u(0)e^{\,m(s,k)\ell(t)},
\qquad
w(t)=w(0)e^{-\,m(s,k)\ell(t)},
\end{equation}
where
\begin{equation}
\ell(t)=\int_0^t \frac{d\tau}{f(\tau)}.
\end{equation}
With
\begin{equation}
x(0;s,0,k)=v_0(s,k)
=
\binom{1-\alpha_0(s,k)}{-(1+\alpha_0(s,k))},
\qquad
\alpha_0(s,k)=\tanh\!\Bigl(\frac{m(s,k)}{\delta}\Bigr),
\end{equation}
we get
\begin{equation}
u(T)=(1-\alpha_0(s,k))e^{\,m(s,k)\ell(T)},
\qquad
w(T)=-(1+\alpha_0(s,k))e^{-\,m(s,k)\ell(T)}.
\end{equation}
Since \(\alpha_T(s,k)=-\alpha_0(s,k)\), substitution into \eqref{eq:F_matching_clean} gives \eqref{eq:F-zero-explicit}. Using
\begin{equation}
\frac{1+\tanh z}{1-\tanh z}=e^{2z},
\end{equation}
the equation \(F(s,0,k)=0\) is equivalent to
\begin{equation}
e^{2m(s,k)\ell(T)}
=
\left(\frac{1+\tanh(m(s,k)/\delta)}{1-\tanh(m(s,k)/\delta)}\right)^2
=
e^{4m(s,k)/\delta},
\end{equation}
which is exactly \eqref{eq:F-zero-factor}. The last statement follows from \eqref{eq:non-degenerate-delta}.
\end{proof}

\begin{corollary}[Transverse boundary zeros are isolated regularized crossings]
\label{cor:transverse-boundary-zeros}
Assume \eqref{eq:non-degenerate-delta}. Let \(s_*\) be a boundary-zero for mode \(k\), and assume \(A'(s_*)\neq 0\). Then \(s_*\) is an isolated simple zero of the scalar function \(s\mapsto F(s,0,k)\). More precisely,
\begin{equation}
\partial_s F(s,0,k)\big|_{s=s_*}
=
2A'(s_*)\Bigl(\ell(T)-\frac{2}{\delta}\Bigr)\neq 0.
\end{equation}
\end{corollary}

\begin{proof}
Differentiate \eqref{eq:F-zero-explicit} with respect to \(s\), and use \(m(s_*,k)=0\) together with \(\partial_s m(s_*,k)=A'(s_*)\).
\end{proof}

\begin{proof}[Proof of Theorem~\ref{thm:section5-main}]
The zero-mode statement is exactly Proposition~\ref{prop:explicit-zero-mode-criterion}. If moreover \(A'(s_*)\neq 0\), then Corollary~\ref{cor:transverse-boundary-zeros} shows that \(s_*\) is an isolated regular crossing. Corollary~\ref{cor:sf-equals-maslov} then places this crossing in the standard modewise spectral-flow/Maslov formalism.
\end{proof}

\begin{proof}[Proof of Corollary~\ref{cor:global-zero-mode-criterion}]
By the orthogonal decomposition \eqref{eq:global-regularized-decomposition},
\begin{equation}
0\in\Spec\bigl(D^{(\delta)}_{s}\bigr)
\quad\Longleftrightarrow\quad
\exists\,k\in\mathcal K \text{ such that } 0\in\Spec\bigl(D^{(\delta)}_{s,k}\bigr).
\end{equation}
The result, therefore, follows from Theorem~\ref{thm:section5-main}.
\end{proof}

\begin{remark}
Theorem~\ref{thm:section5-main} identifies exactly when regularized zero modes occur and when the corresponding crossings are isolated and regular. We do not pursue here a general closed sign formula for arbitrary nontransverse crossings.
\end{remark}

\subsection{Three explicit \(A(s)\)-paths}\label{subsec:three_examples_maslov}

We illustrate the regularized crossing mechanism for three explicit one-parameter gauge families:
\begin{itemize}
\item \textbf{Example 1:} \(A(s)=s-\tfrac12\), \quad \(s\in[0,1]\).
\item \textbf{Example 2:} \(A(s)=3s-\tfrac14\), \quad \(s\in[0,1]\).
\item \textbf{Example 3:} \(A(s)=\sin(4\pi s)\), \quad \(s\in[\varepsilon,1-\varepsilon]\), with \(\varepsilon=\tfrac1{24}\).
\end{itemize}
For these plotted examples, we restrict to \(k\in\mathbb Z\). Since the global problem decomposes into Fourier modes, the figures display the modewise branches corresponding to the finite number of modes that can contribute crossings on the chosen parameter interval.

In each case, the theoretical zero-mode set is determined by
\begin{equation}
k+A(s)=0
\end{equation}
under the non-degenerate assumption. The numerical continuation is illustrative: it visualizes nearby branches \(F(s,\lambda,k)=0\) and distinguishes transverse crossings from nontransverse touching points.

In Examples 1 and 2, every interior boundary-zero is transverse, since \(A'(s)=1\) and \(A'(s)=3\), respectively. In Example 3, nontransverse boundary zeros occur only in the periodic spin structure, for the modes \(k=\pm1\), at
\begin{equation}
s_* \in \left\{ \frac18,\frac38,\frac58,\frac78 \right\}.
\end{equation}
At these points \(A'(s_*)=0\), so they are touching points rather than transverse crossings. The cutoff \(s\in[\varepsilon,1-\varepsilon]\) ensures endpoint invertibility in Example 3.

\subsection{Numerical plots: branch tracking and crossings}\label{subsec:plots_maslov}

Figures~\ref{fig:maslov_example1}-\ref{fig:maslov_example3} display modewise branches \(s\mapsto \lambda(s,k)\) for selected Fourier modes \(k\). A zero-mode crossing occurs exactly when a branch meets \(\lambda=0\), equivalently, when
\begin{equation}
F(s,\lambda,k)=0
\quad\text{with}\quad
\lambda=0.
\end{equation}
Thus, the figures provide a numerical visualization of the criterion proved above.

For branch tracking, we use the Riccati reformulation of the ratio \(r=w/u\), which is numerically more stable than solving \(F(s,\lambda,k)=0\) from scratch at each step, especially near degenerate configurations such as Example~3. This is numerically equivalent to the same boundary-value problem with the same endpoint conditions \eqref{eq:bc_tanh_clean}.

\begin{figure}[t]
  \centering
  \includegraphics[width=.85\linewidth]{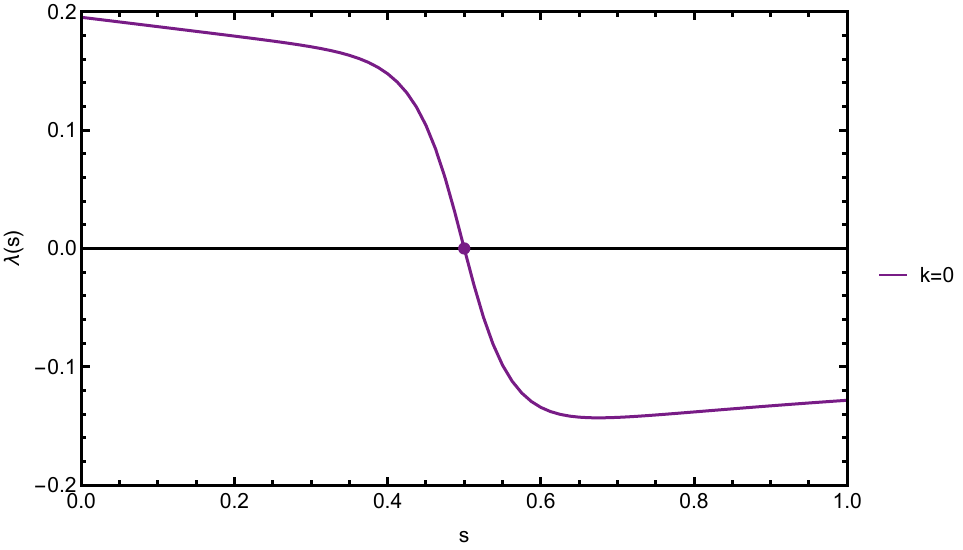}
  \caption{Example 1 (\(A(s)=s-\tfrac12\)): tracked branches \(\lambda(s,k)\) for representative modes \(k\). Boundary-zero locations are marked.}
  \label{fig:maslov_example1}
\end{figure}

\begin{figure}[t]
  \centering
  \includegraphics[width=.85\linewidth]{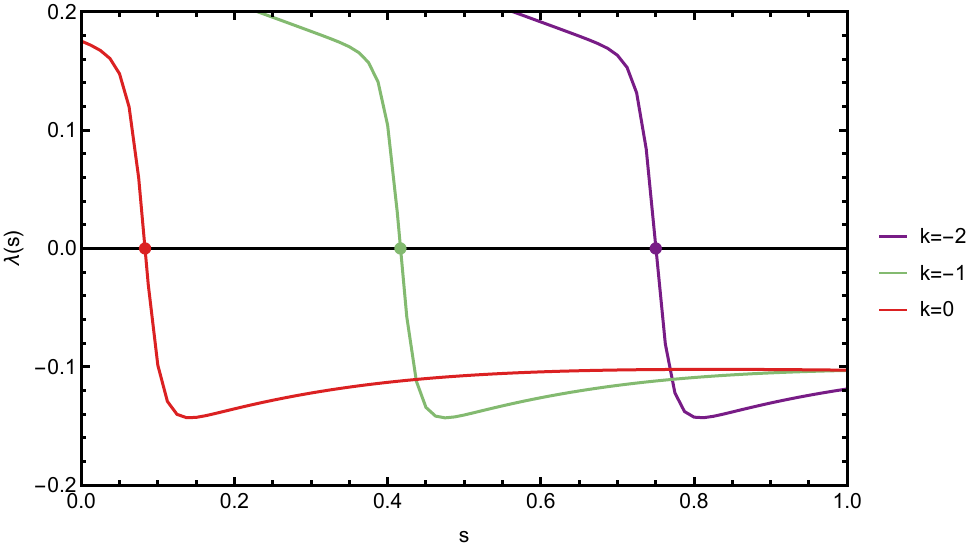}
  \caption{Example 2 (\(A(s)=3s-\tfrac14\)): tracked branches \(\lambda(s,k)\) for selected modes.}
  \label{fig:maslov_example2}
\end{figure}

\begin{figure}[t]
  \centering
  \includegraphics[width=.85\linewidth]{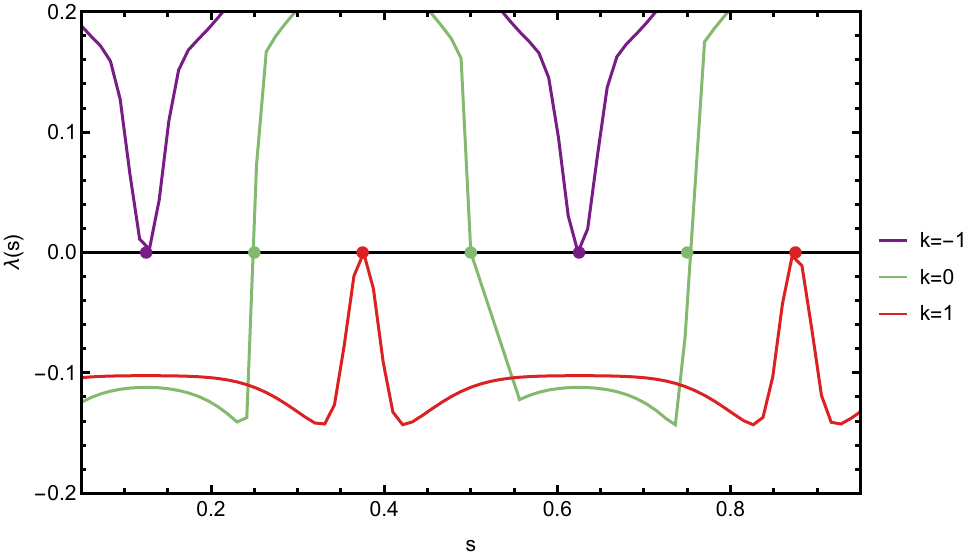}
  \caption{Example 3 (\(A(s)=\sin(4\pi s)\) on \([\varepsilon,1-\varepsilon]\), \(\varepsilon=\tfrac1{24}\)): tracked branches for the selected modes, with endpoint cutoff ensuring invertibility at the ends.}
  \label{fig:maslov_example3}
\end{figure}

The figures are produced numerically from the Riccati reformulation, with the same endpoint conditions as in \eqref{eq:bc_tanh_clean}. Boundary zeros \(m(s,k)=0\) are used as canonical starting points for branch continuation, in agreement with Proposition~\ref{prop:explicit-zero-mode-criterion}.

\paragraph{Conclusion of the regularized analysis.}
Under the non-degenerate assumption, the regularized zero modes occur exactly at \(k+A(s)=0\), and transverse boundary zeros yield isolated regular crossings to which the standard modewise spectral-flow/Maslov formalism applies.

\appendix

\section{Heun reduction of the radial Dirac equation}\label{app:heun}

The main text uses only the following structural facts about the radial mode equation:
(i) after decoupling, it becomes a second-order scalar ODE,
(ii) under the explicit warped function $f(t)=e^t+\alpha e^{-t}$ its coefficients become rational after $z=e^t$,
and (iii) the resulting equation has four \emph{regular singular points}, hence belongs to the general Heun class.
We therefore do \emph{not} need closed-form Heun solutions; what matters for the spectral problem is the singularity
structure and the associated Frobenius exponents.

Throughout, we fix a Fourier mode $k$ and write.
\begin{equation}
m=k+A, \qquad f(t)>0,
\end{equation}
where $f$ is the warping function, $m \neq 0$ and $\lambda\in\mathbb{R}$ is the spectral parameter in $D\psi=\lambda\psi$.

\medskip
\noindent

Our goal in this appendix is to show the Heun reduction in a self-contained way:
\S\ref{app:heun-decoupled} decoupling of the first-order Dirac system to a scalar ODE,
\S\ref{app:heun-liouville} Liouville substitution to remove the first-derivative term,
\S\ref{app:heun-exp}-\S\ref{app:heun-z-sing} rational $z$-equation and identification of its singular points,
\S\ref{app:heun-frob-0}-\S\ref{app:heun-frob-inf} Frobenius exponents at $z=0$, $z=\pm i\sqrt{\alpha}$ and $z=\infty$,
\S\ref{app:heun-classification2}-\S\ref{app:heun-role2} Summary of the Heun classification and its role in the shooting formulation used in Section~3.3.

\noindent
The singularity analysis and extraction of characteristic exponents uses the standard Frobenius method for
second-order linear ODEs with regular singular points, together with the usual reduction of the point
$z=\infty$ via $z=1/\zeta$; see, e.g., Ince~\cite{InceODE}. Finally, the fact that a second-order equation
with four regular singularities belongs to the general Heun class is standard; see Ronveaux~\cite{RonveauxHeun}.

\subsection{Decoupled radial equation}\label{app:heun-decoupled}
We begin by decoupling the first-order radial Dirac system into a single scalar equation
for $u$, which will serve as the master ODE governing each Fourier mode. In \autoref{subsec:bulk-spectrum} we obtained the coupled system
\begin{subequations}
\begin{align}
v' + \frac{f'}{2f}\,v + \frac{m}{f}\,v &= -\,i\lambda\,u, \label{eq:mode-v-app}\\
u' + \frac{f'}{2f}\,u - \frac{m}{f}\,u &= -\,i\lambda\,v. \label{eq:mode-u-app}
\end{align}
\end{subequations}
Eliminating \(v\) yields a second-order equation for \(u\). Differentiating \eqref{eq:mode-u-app} and using
\eqref{eq:mode-v-app} to substitute for \(v\) and \(v'\) gives
\begin{equation}\label{eq:u-second-order}
u'' + \frac{f'}{f}\,u'
+\Bigg[
\lambda^2
- \frac{m^2}{f^2}
+ \frac{m f'}{f^2}
+ \frac{f''}{2f}
- \frac{(f')^2}{4f^2}
\Bigg]u=0.
\end{equation}
This is the master scalar equation governing each Fourier mode of the bulk spectrum.

\subsection{Liouville transformation}\label{app:heun-liouville}
In this section, we remove the first-derivative term by the standard Liouville substitution $u=f^{-1/2}w$; to this end, we apply the Liouville substitution
\begin{equation}\label{eq:liouville-sub}
u(t)=f(t)^{-1/2}\,w(t).
\end{equation}
A direct computation shows that \eqref{eq:u-second-order} is equivalent to the Schr\"odinger-type equation
\begin{equation}\label{eq:w-schrodinger}
w'' + V(t)\,w=0,
\end{equation}
where the effective potential is
\begin{equation}\label{eq:V-effective}
V(t)=\lambda^2-\frac{m^2}{f(t)^2}+ \frac{m f'(t)}{f(t)^2}.
\end{equation}
In particular, \eqref{eq:w-schrodinger} isolates the singular structure coming from the warp factor \(f\).

\subsection{Exponential change of variables}\label{app:heun-exp}
Now we convert the equation into a rational-coefficient ODE via $z=e^t$.
The warp factor for our metric is
\begin{equation}
f(t)=e^{t}+\alpha e^{-t}\qquad (\alpha>0),
\end{equation}
and introduce the exponential coordinate
\begin{equation}\label{eq:z-exp}
z=e^{t},\qquad z\in (1,e^{T}).
\end{equation}
Then
\begin{equation}
f(t)=z+\frac{\alpha}{z}=\frac{z^2+\alpha}{z},\qquad
\frac{d}{dt}=z\frac{d}{dz}.
\end{equation}
Under this change of variables, \eqref{eq:w-schrodinger} becomes
\begin{equation}\label{eq:w-z-form}
\frac{d^2 w}{dz^2}+\frac{1}{z}\frac{dw}{dz}+\frac{1}{z^2}\,Q(z)\,w=0,
\end{equation}
where
\begin{equation}\label{eq:Qz}
Q(z)=
\lambda^2
-\frac{m^2 z^2}{(z^2+\alpha)^2}
+\frac{m(z^2-\alpha)}{(z^2+\alpha)^2}.
\end{equation}

\subsection{The \(z\)-equation and its singular points}\label{app:heun-z-sing}
After the Liouville substitution \(u=f^{-1/2}w\) and the change of variable \(z=e^t\), the mode equation takes the form
\begin{equation}\label{eq:w-z-form-app2}
w_{zz}+\frac{1}{z}w_{z}+\frac{1}{z^{2}}Q(z)\,w=0,
\end{equation}
where \(Q(z)\) is the rational function displayed in \eqref{eq:Qz}. The singular points of
\eqref{eq:w-z-form-app2} are exactly the points where the coefficients fail to be analytic: \(z=0\), the zeros of
\(z^{2}+\alpha\) (namely \(z=\pm i\sqrt{\alpha}\)), and \(z=\infty\). Each of these points is a \emph{regular singular}
point (see \autoref{app:heun-frob-0}-\ref{app:heun-frob-inf} below). Hence \eqref{eq:w-z-form-app2} is a Fuchsian
second-order equation with four regular singularities, which places it in the (general) Heun class.

\medskip\noindent
\emph{Remark.} We will not need the explicit canonical Heun parameters in the main text. What matters for our purposes
is the singularity structure and the fact that the spectral problem can be handled via the shooting formulation in \autoref{subsec:bulk-spectrum} without invoking closed-form Heun solutions.

\subsection{Frobenius analysis at \(z=0\)}\label{app:heun-frob-0}
Here, we discuss the Frobenius exponents at $z=0$, which determine the local behavior of solutions
near the left endpoint $t\to -\infty$ in $z=e^{t}$ coordinate.

As \(z\to 0\) we have \(z^{2}+\alpha=\alpha+O(z^{2})\), hence expanding \(Q(z)\) gives
\begin{equation}
Q(z)=\lambda^{2}- \frac{m}{\alpha}+O(z^{2}).
\end{equation}
Therefore, the leading model near \(z=0\) is
\begin{equation}
w_{zz}+\frac{1}{z}w_{z}+\frac{1}{z^{2}}\Big(\lambda^{2} - \frac{m}{\alpha}\Big)w=0.
\end{equation}
Substituting a Frobenius ansatz \(w\sim z^{\rho}\) yields the indicial equation
\begin{equation}
\rho(\rho-1)+\rho+\Big(\lambda^{2} - \frac{m}{\alpha}\Big)=0
\qquad\Longleftrightarrow\qquad 
\rho^{2}+\Big(\lambda^{2}- \frac{m}{\alpha}\Big)=0.
\end{equation}
Hence, the characteristic exponents at \(z=0\) are
\begin{equation}
\rho_{0}^{\pm}=\pm i\sqrt{\lambda^{2}- \frac{m}{\alpha}}
\qquad \text{when} \ \lambda^{2}- \tfrac{m}{\alpha}>0,
\end{equation}
and 
\begin{equation}
\rho_{0}^{\pm}=\pm \sqrt{\frac{m}{\alpha} - \lambda^{2}} \qquad \text{when} \ \lambda^{2}- \tfrac{m}{\alpha}<0
\end{equation}

\subsection{Frobenius analysis at \(z=\pm i\sqrt{\alpha}\)}\label{app:heun-frob-pm}
Let \(z_{\ast}\in\{+i\sqrt{\alpha},-i\sqrt{\alpha}\}\) and write \(z=z_{\ast}+\zeta\). Since \(z^{2}+\alpha\) has a
simple zero at \(z=z_{\ast}\), we have
\begin{equation}
z^{2}+\alpha = 2z_{\ast}\,\zeta + O(\zeta^{2}),
\end{equation}
We observe \(\frac{1}{z^{2}}Q(z)\) has a double pole at \(\zeta=0\). In particular, each of
\(z=\pm i\sqrt{\alpha}\) is a \emph{regular singular point}. Writing \(w\sim \zeta^{\rho}\) and collecting the most
singular terms yield an indicial equation of the form
\begin{equation}
\rho(\rho-1)+p_{\ast}\rho+q_{\ast}=0,
\end{equation}
where \(p_{\ast},q_{\ast}\) depend explicitly on \((m,\alpha)\) through the principal part of the Laurent expansion of
\(Q(z)\) at \(z=z_{\ast}\). We record the corresponding characteristic exponents as
\begin{equation}
\rho_{\ast}^{\pm}=\frac{1-p_{\ast}}{2}\pm\sqrt{\Big(\frac{1-p_{\ast}}{2}\Big)^{2}-q_{\ast}}.
\end{equation}
(Explicit expressions for \(p_{\ast},q_{\ast}\) can be obtained by a direct expansion of \(Q\) at
\(z=z_{\ast}\), but are not required for the spectral analysis in the main text.)

\subsection{Frobenius analysis at \(z=\infty\)}\label{app:heun-frob-inf}
To analyze \(z=\infty\), set \(z=1/\zeta\) and rewrite \eqref{eq:w-z-form-app2} in the variable \(\zeta\to 0\).
Since \(Q(z)=\lambda^{2}+O(1/z^2)\) as \(z\to\infty\), the leading model is
\begin{equation}
w_{zz}+\frac{1}{z}w_{z}+\frac{\lambda^{2}}{z^{2}}w=0,
\end{equation}
which shows that \(z=\infty\) is again a \emph{regular singular point}. The corresponding asymptotic behaviors are
\begin{equation}
w(z)\sim z^{\pm i\lambda}
\qquad
\qquad z\to\infty
\end{equation}
This confirms that $z=\infty$ is again a regular singular point and that solutions have
power-law/oscillatory behavior controlled by $\lambda$.

\subsection{Heun classification}\label{app:heun-classification2}
Equation \eqref{eq:w-z-form-app2} has four regular singular points
\begin{equation}
z=0,\qquad z=+i\sqrt{\alpha},\qquad z=-i\sqrt{\alpha},\qquad z=\infty.
\end{equation}
A second-order Fuchsian ODE with four regular singular points is equivalent, after a M\"obius change of variable and a
gauge transformation removing characteristic exponents, to a (general) Heun equation. This justifies describing the
mode equation as \emph{Heun-type} and explains why closed-form expressions for the bulk spectrum are not expected in
general.

\subsection{Role of the Heun reduction in the spectral analysis}\label{app:heun-role2}
The Heun classification is used here only to clarify the analytic complexity of the reduced radial equation. Our
spectral analysis does not require explicit Heun solutions: instead, we use the shooting formulation in \autoref{subsec:bulk-spectrum}, where eigenvalues are characterized as zeros of a scalar endpoint condition. The
Frobenius data above remain useful for 
\begin{enumerate}[label=(\roman*)]
\item checking local behavior, 
\item designing stable numerical integration, and 
\item developing asymptotic regimes (e.g.\ large \(|m|\) or large \(T\)).    
\end{enumerate}

\section{Complement: Gorokhovsky--Lesch gauge-conjugation picture}\label{subsec:GLflow}

\paragraph{A complementary gauge-conjugation viewpoint.}
This discussion is logically independent of the regularized APS/Maslov analysis above. Its purpose is to record a complementary gauge-conjugation interpretation for \emph{local elliptic self-adjoint boundary conditions} preserved by the relevant unitary gauge transformation. It should therefore be read as a parallel viewpoint, not as a second proof of the regularized crossing results.

Consider a local elliptic self-adjoint boundary condition
\begin{equation}
\mathcal F=(\mathcal F_0,\mathcal F_T),
\qquad
\mathcal F_0\subset H^{1/2}(Y_0;S_{Y_0}),
\quad
\mathcal F_T\subset H^{1/2}(Y_T;S_{Y_T}),
\end{equation}
such that the associated bulk Dirac operator is self-adjoint and Fredholm. Writing
\begin{equation}
\gamma:H^1(M;S\otimes E)\to
H^{1/2}(Y_0;S_{Y_0})\oplus H^{1/2}(Y_T;S_{Y_T})
\end{equation}
for the trace map, we set
\begin{equation}\label{eq:Dom-DAF}
\Dom(D_{A,\mathcal F})
=
\{\psi\in H^1(M;S\otimes E):\gamma(\psi)\in \mathcal F_0\oplus\mathcal F_T\},
\end{equation}
and denote by
\begin{equation}
D_{A,\mathcal F}:\Dom(D_{A,\mathcal F})\subset L^2(M;S\otimes E)\to L^2(M;S\otimes E)
\end{equation}
the resulting self-adjoint Fredholm operator.

Let
\begin{equation}
\mathfrak g:S^1\to U(1)
\end{equation}
be a smooth unitary, extended trivially in the \(t\)-direction to \(M=[0,T]\times S^1\), and assume that \(\mathcal F\) is \(\mathfrak g\)-invariant:
\begin{equation}
\mathfrak g(\mathcal F_0\oplus\mathcal F_T)=\mathcal F_0\oplus\mathcal F_T.
\end{equation}
Then multiplication by \(\mathfrak g\) preserves \(\Dom(D_{A,\mathcal F})\).

We fix \(A_0\in\mathbb R\). The Gorokhovsky-Lesch gauge-conjugation path is
\begin{equation}\label{eq:GL_path}
D_{A_0,\mathcal F}(\tau)
=
D_{A_0,\mathcal F}
+
\tau\,\mathfrak g^{-1}[D_{A_0},\mathfrak g],
\qquad \tau\in[0,1].
\end{equation}
Since \(\mathfrak g\) depends only on \(\theta\),
\begin{equation}
[D_{A_0},\mathfrak g]
=
i\gamma_2\,\frac1{f(t)}\,\partial_\theta\mathfrak g,
\qquad
\mathfrak g^{-1}[D_{A_0},\mathfrak g]
=
i\gamma_2\,\frac1{f(t)}\,\mathfrak g^{-1}\partial_\theta\mathfrak g,
\end{equation}
so this is a bounded zeroth-order multiplication operator. Hence \eqref{eq:GL_path} is a norm-continuous path of self-adjoint Fredholm operators with common domain \(\Dom(D_{A_0,\mathcal F})\), and
\begin{equation}
D_{A_0,\mathcal F}(1)=\mathfrak g^{-1}D_{A_0,\mathcal F}\mathfrak g.
\end{equation}

Let
\begin{equation}
P_{A_0,\mathcal F}=\mathbf 1_{[0,\infty)}(D_{A_0,\mathcal F}).
\end{equation}
Now we define the Toeplitz operator
\begin{equation}
T_{\mathfrak g}
=
P_{A_0,\mathcal F}\,\mathfrak g\,P_{A_0,\mathcal F}
:
\Ran(P_{A_0,\mathcal F})\to \Ran(P_{A_0,\mathcal F}).
\end{equation}
Under the hypotheses above, \(T_{\mathfrak g}\) is Fredholm, and the Gorokhovsky-Lesch theorem gives
\begin{equation}\label{eq:GL_SF_Toeplitz}
\operatorname{SF}\bigl(\{D_{A_0,\mathcal F}(\tau)\}_{\tau\in[0,1]}\bigr)
=
\ind(T_{\mathfrak g});
\end{equation}
see \cite{GorokhovskyLesch2013}.

We now specialize to
\begin{equation}
\mathfrak g_N(\theta)=e^{iN\theta},
\qquad N\in\mathbb Z,
\end{equation}
whose winding number is
\begin{equation}\label{eq:winding_def}
\wind(\mathfrak g_N)
=
\frac{1}{2\pi i}\int_{S^1}\mathfrak g_N^{-1}\,d\mathfrak g_N
=
N.
\end{equation}
Equivalently, \(\mathfrak g_N\) represents the class \(N\in K^1(S^1)\cong\mathbb Z\).

In our cylinder model,
\begin{equation}
\mathfrak g_N^{-1}(\partial_\theta+iA)\mathfrak g_N
=
\partial_\theta+i(A+N),
\end{equation}
so
\begin{equation}\label{eq:gauge_shift_bulk}
D_{A+N}=\mathfrak g_N^{-1}D_A\mathfrak g_N.
\end{equation}
If \(\mathcal F\) is also \(\mathfrak g_N\)-invariant, then
\begin{equation}\label{eq:gauge_shift_realization}
D_{A+N,\mathcal F}
=
\mathfrak g_N^{-1}D_{A,\mathcal F}\mathfrak g_N.
\end{equation}

Therefore, whenever
\begin{equation}
A(s_2)-A(s_1)=N\in\mathbb Z,
\end{equation}
the endpoint operators \(D_{A(s_1),\mathcal F}\) and \(D_{A(s_2),\mathcal F}\) are gauge-equivalent, and the GL theorem yields
\begin{equation}\label{eq:GL_takeaway}
\operatorname{SF}\bigl(\{D_{A(s_1),\mathcal F}(\tau)\}_{\tau\in[0,1]}\bigr)
=
\ind\!\bigl(P_{A(s_1),\mathcal F}\,\mathfrak g_N\,P_{A(s_1),\mathcal F}\bigr),
\end{equation}
where
\begin{equation}
P_{A(s_1),\mathcal F}=\mathbf 1_{[0,\infty)}(D_{A(s_1),\mathcal F}).
\end{equation}
In the present circle model, this index depends only on the homotopy class of \(\mathfrak g_N\), hence, only on the winding number \(N\).

The GL construction is a \emph{gauge-conjugation} computation for gauge-equivalent endpoint operators under local boundary conditions. It is therefore complementary to, but not literally the same path as, the parameter family \(s\mapsto D_{A(s),\mathcal F}\). In particular, it applies when the endpoint shift is an integer and the chosen local elliptic self-adjoint boundary condition is preserved by the corresponding unitary \(\mathfrak g_N\).

\subsubsection*{Three explicit examples}

\paragraph{Example 1: \(A(s)=s-\tfrac12\), \(s\in[0,1]\).}
Here
\begin{equation}
A(1)-A(0)=\tfrac12-(-\tfrac12)=1.
\end{equation}
Hence one chooses
\begin{equation}
\mathfrak g_1(\theta)=e^{i\theta},
\qquad
\wind(\mathfrak g_1)=1.
\end{equation}
If \(\mathcal F\) is \(\mathfrak g_1\)-invariant, then
\begin{equation}
D_{A(1),\mathcal F}
=
\mathfrak g_1^{-1}D_{A(0),\mathcal F}\mathfrak g_1,
\end{equation}
and
\begin{equation}
\operatorname{SF}\bigl(\{D_{A(0),\mathcal F}(\tau)\}_{\tau\in[0,1]}\bigr)
=
\ind\!\bigl(P_{A(0),\mathcal F}\,\mathfrak g_1\,P_{A(0),\mathcal F}\bigr).
\end{equation}

\paragraph{Example 2: \(A(s)=3s-\tfrac14\), \(s\in[0,1]\).}
Here
\begin{equation}
A(1)-A(0)=\tfrac{11}{4}-\Bigl(-\tfrac14\Bigr)=3.
\end{equation}
Hence one chooses
\begin{equation}
\mathfrak g_3(\theta)=e^{i3\theta},
\qquad
\wind(\mathfrak g_3)=3.
\end{equation}
If \(\mathcal F\) is \(\mathfrak g_3\)-invariant, then
\begin{equation}
D_{A(1),\mathcal F}
=
\mathfrak g_3^{-1}D_{A(0),\mathcal F}\mathfrak g_3,
\end{equation}
and
\begin{equation}
\operatorname{SF}\bigl(\{D_{A(0),\mathcal F}(\tau)\}_{\tau\in[0,1]}\bigr)
=
\ind\!\bigl(P_{A(0),\mathcal F}\,\mathfrak g_3\,P_{A(0),\mathcal F}\bigr).
\end{equation}

\paragraph{Example 3: \(A(s)=\sin(4\pi s)\), \(s\in[\varepsilon,1-\varepsilon]\), \(\varepsilon=\tfrac1{24}\).}
The cutoff excludes endpoint boundary zeros and therefore guarantees endpoint invertibility. In this case,
\begin{equation}
A(1-\varepsilon)-A(\varepsilon)
=
\sin(4\pi-4\pi\varepsilon)-\sin(4\pi\varepsilon)
=
-2\sin(4\pi\varepsilon)
=
-1.
\end{equation}
Hence one chooses
\begin{equation}
\mathfrak g_{-1}(\theta)=e^{-i\theta},
\qquad
\wind(\mathfrak g_{-1})=-1.
\end{equation}
If \(\mathcal F\) is \(\mathfrak g_{-1}\)-invariant, then
\begin{equation}
D_{A(1-\varepsilon),\mathcal F}
=
\mathfrak g_{-1}^{-1}D_{A(\varepsilon),\mathcal F}\mathfrak g_{-1},
\end{equation}
and
\begin{equation}
\operatorname{SF}\bigl(\{D_{A(\varepsilon),\mathcal F}(\tau)\}_{\tau\in[0,1]}\bigr)
=
\ind\!\bigl(P_{A(\varepsilon),\mathcal F}\,\mathfrak g_{-1}\,P_{A(\varepsilon),\mathcal F}\bigr).
\end{equation}

\section{Numerical implementation and plots}\label{sec:numerics}

\begin{proposition}[Liouville/Wronskian identity]\label{prop:liouville}
For any fundamental matrix $\psi(t)$ we have
\begin{equation}
f(t)\,\det\psi(t)=\textup{constant in }t.
\end{equation}
Equivalently, for any two solutions $\psi^{(1)}=(u_1,v_1)$ and $\psi^{(2)}=(u_2,v_2)$,
\begin{equation}
f(t)\,\bigl(u_1(t)v_2(t)-u_2(t)v_1(t)\bigr)
\end{equation}
is $t$-independent.
\end{proposition}

\begin{proof}
Liouville's formula gives
\begin{equation}
\frac{d}{dt}\det\psi(t)=\mathrm{tr}\bigl(A(t;\lambda)\bigr)\det\psi(t).
\end{equation}
In our system,
\begin{equation}
\mathrm{tr}\bigl(A(t;\lambda)\bigr)=(-q+p)+(-q-p)=-2q(t)=-\frac{f'(t)}{f(t)}.
\end{equation}
Thus, $\frac{d}{dt}\log\det\psi(t)=-\frac{f'(t)}{f(t)}$, hence $\det\psi(t)=C/f(t)$ and $f(t)\det\psi(t)=C$ is constant.
\end{proof}

\paragraph{Numerical illustration of Proposition~\ref{prop:liouville}.}
To complement \ref{prop:liouville}, we plot in \autoref{fig:appendixC-numerics} the numerical deviation of the conserved quantity
\begin{equation}
f(t)\det\psi(t)
\end{equation}
from its exact constant value. Here \(\psi(t)\) is the fundamental matrix of the first-order system, normalized by \(\psi(0)=I\), so the exact constant is
\begin{equation}
f(0)\det\psi(0)=f(0)=2
\end{equation}
for the parameter choice \(\alpha=1\). To match the bulk-spectrum example of \autoref{subsec:bulk_spectrum}, we use
\begin{equation}
\alpha=1,\qquad A=0.3,\qquad k=1,\qquad T=1.5,
\end{equation}
and fix the representative spectral value \(\lambda=2\). Since the exact quantity is constant, the plotted oscillations represent only numerical error. For visibility, we therefore plot the scaled error
\begin{equation}
10^{8}\,\bigl|f(t)\det\psi(t)-2\bigr|.
\end{equation}
We also include a zoomed view over a shorter subinterval to display the fine oscillatory structure more clearly.

\begin{figure}[t]
    \centering
    \includegraphics[width=0.48\textwidth]{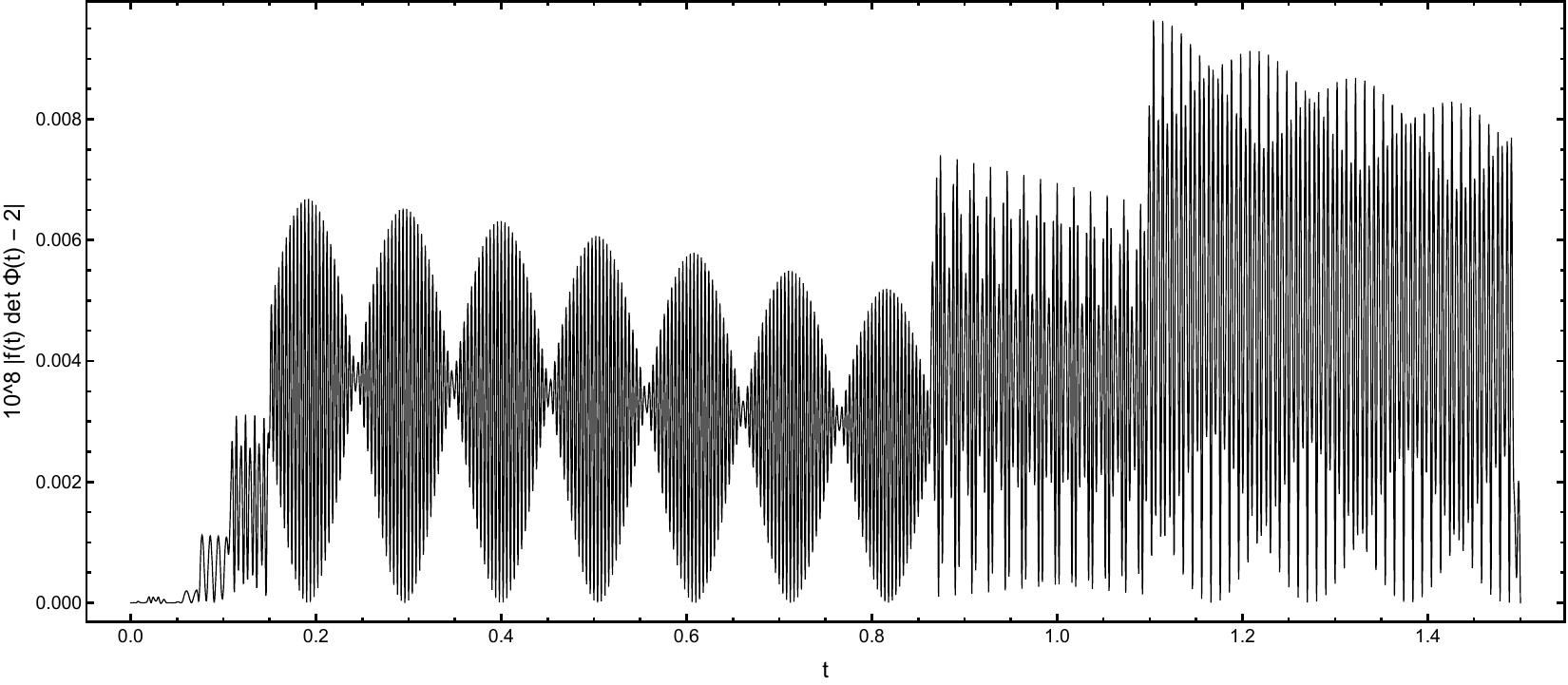}\hfill
    \includegraphics[width=0.48\textwidth]{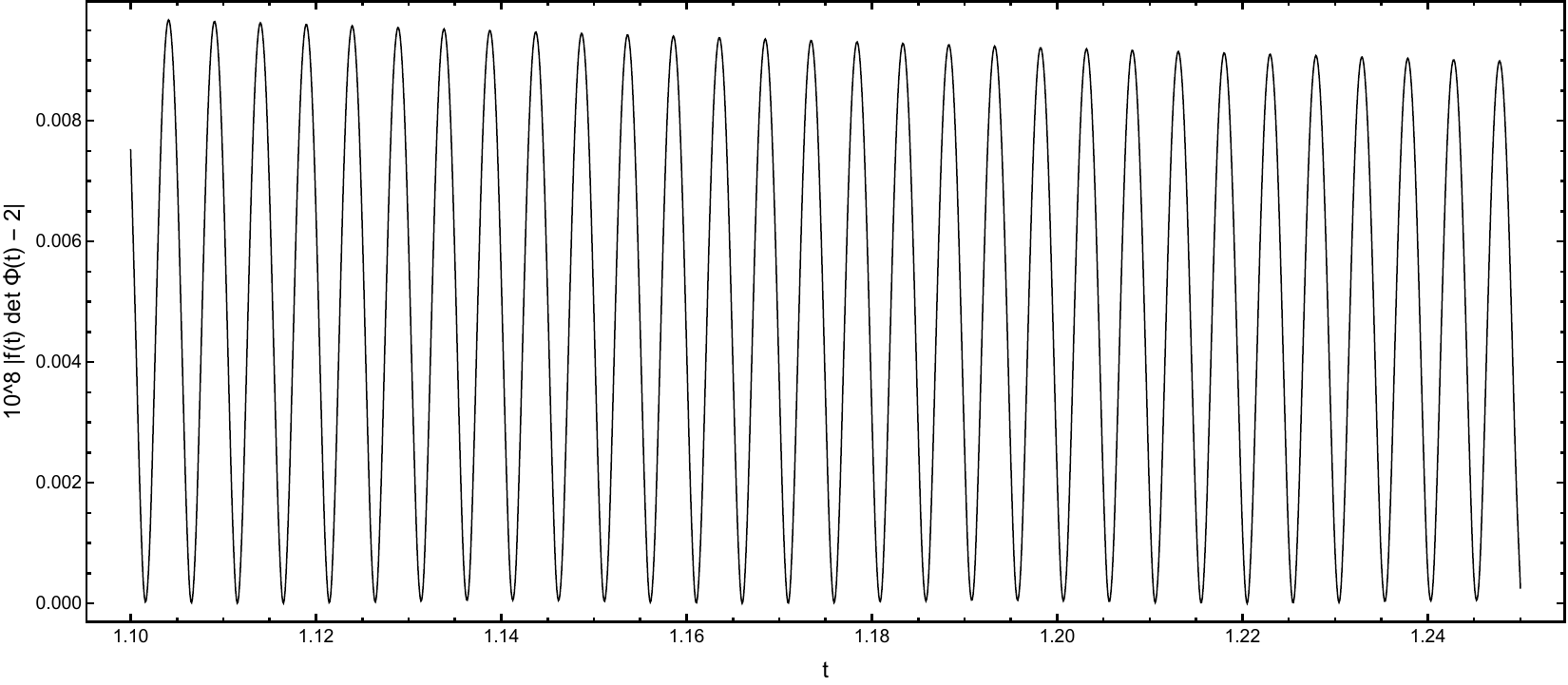}
    \caption{Numerical illustration of Proposition~\ref{prop:liouville} for the parameter values $\alpha=1$, $A=0.3$, $k=1$, and $T=1.5$, matching the setup of \autoref{subsec:bulk_spectrum}, with representative spectral parameter $\lambda=2$. The plotted quantity is $10^{8}|f(t)\det\psi(t)-2|$, where $\psi(t)$ is the fundamental matrix normalized by $\psi(0)=I$. Since the exact conserved value is $f(0)=2$, the graph measures only numerical deviation from the Liouville/Wronskian identity. Right: zoomed view on a shorter subinterval showing the oscillatory structure more clearly.}
    \label{fig:appendixC-numerics}
\end{figure}

%% References
\bibliographystyle{utphys}
\bibliography{ref}

\end{document}